\documentclass[a4paper]{article}
\usepackage{authblk}

\usepackage{microtype}

\title{Access structure in graphs in high dimension and application to secret sharing}

\author[1]{Anne Marin\footnote{anne.marin@telecom-paristech.fr}}
\author[2,1]{Damian Markham\footnote{damian.markham@telecom-paristech.fr}}
\author[2,3]{Simon Perdrix\footnote{Simon.Perdrix@imag.fr}}
\affil[1]{\small LTCI, INFRES, Telecom ParisTech, France}
\affil[2] {CNRS}
\affil[3]{LIG, Grenoble University, France}
\usepackage{amsmath}
\usepackage{setspace}
\usepackage{array}
\usepackage{times}
\usepackage{amssymb}
\usepackage{amstext}
\usepackage{amsfonts}
\usepackage{mathrsfs}
\usepackage[all]{xy}
\usepackage{pgf}
\usepackage{svgcolor}
\usepackage{tikz}
\usepackage{rotating}
\usepackage{amsthm}
\usepackage{version}
\usepackage{placeins}
\usepackage[retainorgcmds]{IEEEtrantools}
\usepackage{dsfont}
\usepackage{subfig}

\newcommand{\ket}[1]{|#1\rangle}
\newcommand{\bra}[1]{\langle#1|}

\newcommand{\C}{\mathbb {C}}
\newcommand{\F}{\mathbb {F}}

\newcommand{\cutrk}{\textup{cutrk}}
\newcommand{\prk}[1]{\textup{rk}_{#1}}
\newcommand{\rk}{\textup{rank}}
\newcommand{\supp}{\textup{sup}}

\newtheorem{theorem}{Theorem}
\newtheorem{lemma}{Lemma}
\newtheorem{definition}{Definition}
\newtheorem{corollary}{Corollary}
\newtheorem{example}{Example}

\date{}

\begin{document}

\maketitle

\begin{abstract}
We give graphical characterisation of the access structure to both classical and quantum information encoded onto a multigraph defined for prime dimension $q$, as well as explicit decoding operations for quantum secret sharing based on graph state protocols. We give a lower bound on $k$ for the existence of a $((k,n))_q$ scheme and prove, using probabilistic methods, that there exists $\alpha$ such that a random multigraph has an accessing parameter $k\leq \alpha n$ with high probability.
\end{abstract}

\section{Introduction}
In this work we consider encoding, and accessing, both quantum and classical information onto graph states of qudits - multipartite entangled states which are one to one corresponding to multigraphs (which we will consider as simple graphs with multiple edges). We are particularly interested in using these states for secret sharing.

Secret sharing is an important cryptographic primitive, which was first put forward classically in \cite{S}, and then extended to the quantum realm in \cite{HBB99, CGL99}. The aim of the protocol is for a dealer to distribute a secret (quantum or classical) to a set of players, in such a way that only authorized sets of players can access the secret, and unauthorized sets of players cannot (there may be sets of players which are neither authorized nor unauthorized). The sets of authorized and unauthorized players is called the access structure. Any secret sharing scheme of $n$ players can be loosely paramaterised by two numbers, $k$ and $k'$, such that any subset of $k$ players is an authorized set, whereas any subset of $k'$ players or less is unauthorized. We call such paramaterised schemes $(k,k',n)$ ramp schemes. In the case when $k'=k-1$, we say it is a threshold scheme, and simplify the notation to $(k,n)$.

In this work we consider two classes of quantum schemes, one class using quantum channels to distribute classical secrets, denoted CQ schemes \cite{HBB99}, and the other sharing quantum secrets \cite{CGL99,HBB99}, denoted QQ schemes. The notation CQ and QQ used here follows the work \cite{MS08,Keet10,MM12}, where both classes were phrased in the same language of graph states (first for qubits \cite{MS08} then qudits \cite{Keet10,MM12}). The equivalence of both schemes was shown in \cite{MM12}. Using the graph state formalism can be useful both practically - since graph states are amongst the most well developed multipartite entangled states experimentally - and theoretically, since graph states are rich in their uses in quantum information, and allow for graphical characterization of information flow, and access of information. The connection between error correction and secret sharing was understood early on \cite{CGL99}, and implies that for general access structures it is necessary to use high dimensional states to encode the secret \cite{MS08,MM12}. In \cite{KMMP09} an entirely graphical description of the access structure was given for the graph state protocols on qubits. This has led to many applications, for instance in proving lower and upper bounds on what is $k$ and $k'$ are possible in ramp schemes. We are naturally interested in doing the same for higher dimensional versions. 

The first result of this paper is to extend to higher dimension the characterisation of the access structure in a graph, previously done  in \cite{KMMP09} for 2-dimensional system. By gathering the graphical conditions and previous results, we show that the accessibility problem to quantum information can be reduced to study the classical information's one in both a set of player and its complementary (which was proved in \cite{KMMP09,JMP11} for $2$ dimensional system). Finally we use this result for the decoding phase of both QQ and CQ protocols, as we know \cite{MM12} that a CQ authorised is a QQ authorised set and vice versa. In the last part, we study the existence, as a function of $k$, of a $((k,n))_q$ scheme (this will be defined explicitly later, but can be understood as the underlying graph encoding which gives rise to $(k,n-k,n)$ QQ secret sharing schemes). We derive a lower bound over $k$, that is, there exists $\alpha$ such that every $(k,n-k,n)$ QQ secret sharing must satisfies $k> \alpha n$, and we use probabilistic method to find $c<1$ such that a $((cn,n))_q$ scheme exists with high probability.

\section{Background}

\subsection{Qudit graph states, $\mathbb F_q^*$-graphs, and multigraphs}

The \emph{qudit graph state} formalism \cite{BB06, KKKS05} consists of representing a quantum state using a weighted undirected graph where every vertex represents a $q$-dimensional quantum system and every edge, which has assigned an element from the finite field $\mathbb F_q$, represents intuitively the entanglement between the elementary systems (a formal definition is given in Definition \ref{def:mG}). Such graphs, labeled with elements of a finite field $\mathbb F_q$, are known as \emph{$\mathbb F_q^*$-graphs} \cite{KanteRao} and can be  can be interpreted as edge-colored graphs.  In this paper, we choose to interpret $\mathbb F_q^*$-graphs as multigraphs i.e., graphs with possibly parallel edges between  pairs of vertices. Albeit equivalent to the other interpretation of $\mathbb F_q^*$-graphs, we believe that the multigraph interpretation is relevant in the context of qudit graph states for secret sharing protocols, in particular for the graphical characterisation  of authorised and unauthorised sets of players (see Lemmas \ref{lem:suff_acc} and \ref{lem:suff_no}).

\begin{definition}[$q$-multigraphs] Given a prime number $q$, a $q$-multigraph $G$ is a pair $(V,\Gamma)$ where $V$ is the finite  set of vertices and $\Gamma : V\times V\to \mathbb F_q$ is  the adjacency matrix of $G$: for any $u,v\in V$, $\Gamma(u,v)$ is the multiplicity of the edge $(u,v)$ in $G$.
\label{def:mG}
\end{definition}

The term multigraph is used for $q$-multigraph when $q$ is clear from the context or irrelevant. 
In this paper, we  consider undirected simple multigraphs $G=(V,\Gamma)$ i.e., for any vertices $u,v\in V$, $\Gamma(u,v)=\Gamma(v,u)$ and $\Gamma(u,u)=0$. For our characterizations of encoding and accessing later on, it will be useful to introduce further concepts. We will see several examples of them along the way, but for now we state definitions.
Given a set $V$ of vertices, a vector $D:V\to  {\mathbb F^q}$ represents a multiset of vertices of $V$: for every $v\in V$, $D(v)\in \mathbb F_q$ is the multiplicity of $v$ in $D$. $\supp(D) = \{v\in V~|~D(v)\neq 0\bmod q\}$ is the support of $D$. For any multigraph $G = (V,\Gamma)$ and any multiset of vertices $D:V\to \mathbb F_q$, the matrix product  $\Gamma.D$ is the multiset of neighbours of $D$: for any $v\in V$, $v$ is a neighbour of $D$ with multiplicity $(\Gamma.D) (v) =  \sum_{u\in V} \Gamma(u,v).D(u) \bmod q$. 
In particular, for any vertex $u$, $\Gamma.\{u\}$ is the multiset of neighbours of $u$. We call $G[D]=(V',\Gamma')$  the sub-multigraph of $G=(V,\Gamma)$ induced by the multiset $D:V\to \F_q$, where $V' = V\cap \supp(D)$ and $\Gamma' :V'\times V'\to \mathbb F_q= (u,v) \mapsto D(u).\Gamma(u,v) .D(v) \bmod q$. Notice that the multiplicity of an edge in $G[D]$ is the multiplicity of this edge in the original graph $G$ times the multiplicity in $D$ of the two vertices connected by this edge.   For any $A,B\subseteq V$, $\Gamma[A,B]$  denotes the submatrix of $\Gamma$ whose columns correspond to the vertices in $A$ and  rows to the vertices in $B$.  $\Gamma[A,B]$ represents the edges which have one end in $A$ and the other one in $B$.

\begin{definition}[Qudit Graph State] Given a  $q$-multigraph $G =(V,\Gamma)$ with $V = \{v_1, \ldots , v_n\}$, let $\ket{G}\in \mathbb C^{q^{n}}$ be its associated qudit graph state defined as $$\ket G =\frac1{\sqrt {q^n}} \sum_{x=(x_1,\cdots,x_n)\in \F^n_q} \omega^{|G[x]|}\ket x$$ where $\omega$ is the $q^{\textrm{th}}$  root of unity and  $|G[x]|$ is the number of edges of the sub-multigraph  $G[x] = (V_x,\Gamma_x)$ induced by $x$, where $V_x = \{v_i\in V, x_i\neq 0\}$ and $\Gamma_x : V_x\times V_x \to \mathbb F_q = (v_i,v_j) \mapsto x_ix_j\Gamma(v_i,v_j)$.\\

\end{definition}

Qudit graph states satisfy the following fundamental fixpoint property. Given a $q$-multigraph $G=(V,\Gamma)$, $\ket G$ is the unique quantum state (up to a global phase) such that, for any $u\in V$, \begin{IEEEeqnarray}{C}\label{Eqn: stab} X_uZ_{\Gamma.\{u\}} \ket G = \ket G\label{S}\end{IEEEeqnarray} where 
$\Gamma.\{u\}$ is the multiset of neighbours of $u$, 
 $X = \ket b \mapsto \ket{b+1 \bmod q}$, $Z= \ket b \mapsto \omega^b\ket b$ are Pauli operators, and for any mulitset $D:V\to \mathbb F_q$, $Z_D := \bigotimes_{v\in V}Z_v^{D(v)}$.

\begin{example}We define the $3$-multigraph $G=(V,\Gamma)$ by $V=\{v_1,v_2, v_3, v_4, v_5\}$,\\
\begin{minipage}[c]{0.5\textwidth}
\center$\Gamma=\left[\begin{array}{ccccc} 0&0&1&0&1\\0&0&2&0&1\\1&2&0&2&0\\0&0&2&0&2\\1&1&0&2&0\end{array}\right]$\\
\end{minipage}
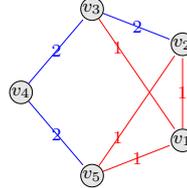
\begin{minipage}[c]{0.5\textwidth}
\scalebox{0.85}{\begin{tikzpicture}[shorten >=1, -, font=\footnotesize]
\tikzstyle{vertex}=[circle,draw,fill=black!10,minimum size=10pt,inner sep=0pt,font=\footnotesize]

\node[vertex] (d) at (0.2,1.5) {$v_4$};
\foreach \name/\y/\text in {P-1/0.2/v_5, P-3/2.8/v_3}
  \node[vertex,xshift=0cm,yshift=0cm] (\name) at (1.3,\y) {$\text$};

\foreach \name/\y/\text in {Q-1/0.75/v_1, Q-2/2.25/v_2}
  \node[vertex,xshift=0cm,yshift=0cm] (\name) at (2.7,\y) {$\text$};

\foreach \to in {1,3}
           {\draw[blue] (d) -- node {$2$}(P-\to);}
\foreach \from/\to in {1/1}
           {\draw[red] (P-\from) -- node {$1$}(Q-\to);}
\foreach \from/\to in {1/2}
          {\draw[red] (P-\from)  node[xshift=0.4cm,yshift=0.6cm] {$1$} -- (Q-\to);}
\foreach \from/\to in {3/1}
           {\draw[red] (P-\from)  node[xshift=0.4cm,yshift=-0.6cm] {$1$} -- (Q-\to);}
\draw[blue] (P-3) -- node {$2$} (Q-2);
\draw[red] (Q-1) -- node {$1$} (Q-2);

\end{tikzpicture}}
\end{minipage}
Let $A=\{v_1,v_2\}$ be a subset of $V$, and $D:A\rightarrow \F_3$ a multiset such that $D(v_1)=2$, $D(v_2)=1$. That is $D=\{v_1,v_1,v_2\}$. 
Then, with previous definitions, the graph induced by $D$ is $G[D]=$ 
\begin{minipage}{0.6cm}
\scalebox{0.85}{\begin{tikzpicture}[shorten >=1, -, font=\footnotesize]
\tikzstyle{vertex}=[circle,draw,fill=black!10,minimum size=10pt,inner sep=0pt,font=\footnotesize]
\foreach \name/\y/\text in {Q-1/0.75/v_1, Q-2/2.25/v_2}
  \node[vertex,xshift=0cm,yshift=0cm] (\name) at (2.7,\y) {$\text$};
\draw[blue] (Q-1) -- node {$2$} (Q-2);
\end{tikzpicture}}
\end{minipage}
The multiset of neighbours of $A$ is  $\{v_1,v_2,v_5, v_5\}$. 
%
The multiset of neighbours of $D$ is $\{v_1,v_2,v_2,v_3\}$. 
%
%

\end{example}

\subsection{Local complementation and cut rank}

The \emph{local complementation} \cite{B94} is a graph transformation which is incredibly useful for the study of graph states \cite{VdN04}. 
Indeed, if two graphs $G$ and $G'$  are locally equivalent (i.e. one can transform $G$ into $G'$ by means of a series of local complementations), they represent the same entanglement (i.e. there exists a local unitary transformation $U$ such that $\ket {G'}= U\ket G$) \cite{VdN04}. Local complementation is extended to multigraphs as follows \cite{KanteRao}:
Given a $q$-multigraph $G=(V,\Gamma)$, $u\in V$ and $\lambda \in \F_q$, the \emph{$\lambda$-local complementation at $u$ of $G$} is the $q$-multigraph ${G{\star^\lambda} u}=(V,\Gamma')$ such that $\forall v,w\in V$, $v\neq w$, $\Gamma'(v,w) = \Gamma(v,w)+\lambda. \Gamma(v,u).\Gamma(u,w) \bmod q$. Keet et al.~\cite{Keet10} have proved that for any $q$-multigraph $G=(V,\Gamma)$, any $u\in V$ and any $\lambda\in \F_q$,  there exists a local unitary transformation $U$ such that $\ket {G{\star^\lambda}u} = U\ket{G}$.

The \emph{cut rank} \cite{OumSeymour} 
is a set function which associates with every set $B$ of vertices the rank of the matrix describing the edges of the cut ($B$, $V{{{\setminus}}} B$): 
Given a multigraph $G=(V,\Gamma)$, let $\Gamma[B]:=\Gamma[B,V{{{\setminus}}} B]$ be the cut matrix of the cut $(B,V{{{\setminus}}} B)$, moreover for any $A,B\subseteq V$, let $\prk{G}(A,B) := \rk(\Gamma{[A,B]})$ and  $\cutrk_G(B) := \prk{G}(B,V{{\setminus}} B)$ be the cut rank of $B$. Notice that $\prk{G}(A,B) = \prk{G}(B,A)$ and $\cutrk_G(B) = \cutrk_G(V{{{\setminus}}} B)$.

 We point out in this paper that the cut rank, which is known to be invariant by local complementation \cite{KanteRao}, is a key parameter of $q$-multigraphs for the study of secret sharing protocols with qudit graph states. Indeed, Theorem \ref{thm:q_acc} states that the capability of a set of players to reconstruct a quantum secret is characterised by the discrete derivative of the cut rank function. Notice that the cut-rank of a bipartition is nothing but the Schmidt measure of entanglement of this bipartition in the corresponding graph state. This is shown for the qubit case in \cite{HEB03}, and easily extends to the qudit case. As a consequence, Theorem \ref{thm:q_acc} characterises the accessibility of a set of players as the derivative of the Schmidt measure of entanglement. 



\subsection{Description of the encoding:} \label{Section: encoding}

We now introduce the encoding of classical and quantum information onto graph states (CQ and QQ respectively), which will be the starting point for the secret sharing protocols defined in section \ref{Section: secret sharing}. For ease of notation we present the CQ encoding as deterministic, and in one basis. When used in the full CQ protocol this is randomised by measurement and choice of basis (described fully in section \ref{Section: secret sharing}). The ability of players to access encoded information (both classical and quantum) is fully described in graph theoretical language in section \ref{Section: access structure}.\\

\textbf{CQ encoding:}\ \\
Given a multigraph $G = (V,\Gamma)$ of order $n$ and a distinguished non isolated vertex $d\in V$, the corresponding CQ encoding of a classical secret  $s\in \F_q$ among $n-1$ players consists of the dealer preparing the state
\begin{IEEEeqnarray*}{C}
\ket {s_L} := Z^s_{\Gamma.\{d\}}\ket {G{\setminus}d },
\end{IEEEeqnarray*}
and sending one qudit to each player, where $G{\setminus}d = (V{\setminus} \{d\}, \Gamma[V{\setminus} \{d\},V{\setminus} \{d\}])$ is the multigraph obtained by removing the vertex $d$ and all its incident edges. 

In the CQ protocol (described in section \ref{Section: secret sharing}) the secret s is randomised by measurement on the dealer's vertex $d$ of the full graph state $|G\rangle$, and further, the encoding is randomised by choice of measurement basis - the dealer chooses at random $t\in T$, $T\subseteq \F_q$ and $|T|\geq 2$, and measures his qudit in the associated complementary basis $X^tZ$. Measuring in this $t$ basis will correspond exactly to using the above CQ encoding of the same secret value $s$ onto the complementary multigraph ${G{\star^t} d}$.\\

\textbf{QQ encoding:}

Given a multigraph $G = (V,\Gamma)$ of order $n$ and a distinguished non isolated vertex $d\in V$, the corresponding QQ encoding on a qudit graph state for sharing an arbitrary quantum secret $\ket \phi = \sum_{j=0}^{q-1} s_j\ket j\in \mathbb C^q$ among $n-1$ players consists, for the dealer, in preparing the state
$$\ket {\phi_L} = \sum_{j=0}^{q-1} s_jZ^j_{\Gamma.\{d\}}\ket {G{\setminus} d}=\sum_{j=0}^{q-1} s_j\ket{j_L}\nonumber
$$ and in sending one qudit of $\ket {\phi_L} $ to each player.

Notice that the preparation consists in applying the map $\ket j \mapsto Z_{\Gamma.\{d\}}^j\ket {G{\setminus} d}$ which is an isometry as long as $d$ is not an isolated vertex in $G$. We describe encoding procedures in appendix \ref{Section:appendix}.

\bigskip

The accessing structure of the protocols (i.e.~the description of the sets of players which can recover the secret, as well as those which have no information about the secret) is given in the next section which provides a graphical characterisation of the accessing structure for the secret sharing protocols using these encodings. Moreover, the operations the authorised sets of players  have to perform to reconstruct the secret are also described in the next section.

\section{Access Structure in a Graph in Higher Dimension:} \label{Section: access structure}
\subsection{Classical Information}

In this section, we show, when the secret is classical, that 
 the protocol is perfect (i.e.~every set of players is either able to recover the secret or has no information about the secret), and that the accessing structure is graphically characterised by a simple  rank-based function:

\begin{theorem}\label{thm:class_char} Given a $q$-multigraph $G=(V,\Gamma)$ and a distinguished vertex $d\in V$, a set $B\subseteq V{{\setminus}} \{d\}$ of players can recover a classical secret for the corresponding CQ encoding if and only if $\pi_G(B,d) =1$, where
 $$\pi_G(B,d):= \cutrk_G(B)-\cutrk_{G\setminus d}(B)$$
 \end{theorem}

A graphical interpretation of Theorem \ref{thm:class_char} is that a set $B$ is accessible if and only if the presence of the `dealer vertex' $d$ increases the rank of the cut between $B$ and the rest of the vertices.

The rest of the section is dedicated to the proof of Theorem \ref{thm:class_char}.

First, we prove that a set $B$ of players can recover a classical secret if, roughly speaking,  there exists a multiset $D$ of them which is not `seen' from outside except by the `dealer':

\begin{lemma}\label{lem:suff_acc} Given a $q$-multigraph $G = (V,\Gamma)$ and $d\in V$, a set
$B\subseteq V{{\setminus}} \{d\}$ of players can recover a classical secret for the corresponding CQ encoding if there exists a multiset $D:B\rightarrow \F_q$ such that $\supp(\Gamma[B,V{{\setminus}} B].D) = \{d\}$ i.e.,
\begin{itemize}
\item the number of neighbours of $d$  in $D$ is not congruent to $0\bmod q$;
\item $\forall u\in V{{\setminus}} (B{\cup} \{d\})$, the number of neighbours of $u$ in $D$ is congruent to $0\bmod q$.
\end{itemize}
\end{lemma}

\begin{proof}
Given  $B \subseteq V$ and $D: B\to \mathbb F_q$ such that $\supp(\Gamma[B,V{{\setminus}} B].D) = \{d\}$. W.l.o.g. we assume  the multiplicity of $d$ in $\Gamma.D$ is $1$ (otherwise we consider the multiset $D'=u\mapsto (\Gamma.D)(d)^{-1}.D(u)$ instead of $D$). The players in $B$ can recover the secret by measuring an appropriate product of stabilizers. Indeed, there exists $r\in \mathbb F_q$ such that $\prod_{u\in B} (X_uZ_{\Gamma.\{u\}})^{D(u)} = \omega^{r}X_D Z_{\Gamma.D}=Z_d\omega^r X_DZ_{\Gamma[V, V{\setminus} \{d\}].D}$.  
As $\prod_{u\in B} (X_uZ_{\Gamma.\{u\}})^{D(u)}\ket G=\ket G$, we deduce that $\omega^r X_D Z_{\Gamma[V, V{{\setminus}} \{d\}].D} \ket{G{{\setminus}} d} = \ket {G{{\setminus}} d}$. 
If the classical secret is $s\in \mathbb F_q$, \\ 
$\omega^r X_DZ_{\Gamma[V, V{{\setminus}} \{d\}].D} Z^s_{\Gamma.\{d\}} \ket {G{{\setminus}} d}= \omega^{r-s}Z^s_{\Gamma.\{d\}}X_DZ_{\Gamma[V, V{{\setminus}} \{d\}].D} \ket {G{{\setminus}} d}=  \omega^{-s}Z^s_{\Gamma.\{d\}} \ket {G{{\setminus}} d}$. So if the players in $B$ measure according to $\omega^r X_DZ_{\Gamma[V, V{{\setminus}} \{d\}].D} $, they get the outcome $-s\bmod q$, so  they recover the classical secret $s$. \end{proof}

Lemma \ref{lem:suff_acc} provides a sufficient condition for a set of players to be able to reconstruct a classical secret. Notice that this reconstruction is nothing but a Pauli measurement, so it can be done by means of local Pauli measurements and classical communications.

\begin{corollary}\label{cor:suff_acc}
Given a $q$-multigraph $G = (V,\Gamma)$, $d\in V$, and $B\subseteq V{{\setminus}} \{d\}$, if $\pi_G(B,d) =1$ then $B$ can reconstruct a classical secret for the corresponding CQ encoding.
\end{corollary}

\begin{proof}Let $F = V{{\setminus}} (B{\cup} \{d\})$. According to lemma \ref{lem:suff_acc}, $B$ can recover a classical secret if there exists $D:B\rightarrow \F_q$ such that $\supp(\Gamma[B,V{{\setminus}} B].D) = \{d\}$. W.l.o.g. we can assume that the multiplicity of  $d$ in $\Gamma[B,V{{\setminus}} B].D$ is one. So $B$ can recover a classical secret if the system $\left(\begin{array}{l}\Gamma[B,\{d\}]\\ \hline \Gamma[B,F]\end{array}\right).x=\left(\begin{array}{l}1\\ \hline 0\end{array}\right)$
has a non zero solution, which is equivalent to $\rk\left(\begin{array}{l}\Gamma[B,\{d\}]\\ \hline \Gamma[B,F]\end{array}\right)=\rk\left(\begin{array}{c|c}\Gamma[B,\{d\}] & 1\\ \hline \Gamma[B,F]& 0\end{array}\right)$. 
Using the last column of the right-side matrix to cancel terms of the row $\Gamma[B,\{d\}]$, we are finally reduced to $\rk\left(\begin{array}{l}\Gamma[B,\{d\}]\\ \hline \Gamma[B,F]\end{array}\right)=1+\rk(\Gamma[B,F])$ i.e.,  $ \cutrk_G(B)-\prk{G}(B, F) =1=\pi_G(B,d)$.
\end{proof}

In the following,  a sufficient condition for a set of players to have no information about the secret is introduced: roughly speaking, a multiset of players $D$  which includes the dealer $d$, can `hide' the secret to every player who is connected to $D$ with a number of edges congruent to $0$ modulo $q$:

\begin{lemma}\label{lem:suff_no}
Given a $q$-multigraph $G = (V,\Gamma)$ and $d\in V$, a set
 $B\subseteq V{{\setminus}} \{d\}$ has no information about a classical secret for the corresponding CQ encoding if there exists $D:V{{\setminus}} B \rightarrow \F_q$, such that $D(d)\neq 0\bmod q$ and $\Gamma[V{{\setminus}} B, B].D=0$ i.e.,
\begin{itemize}
\item the multiplicity of $d$ in $D$ is not congruent to $0\bmod q$;
\item $\forall u\in B$, the number of neighbours of $u$ in $D$ is congruent to $0\bmod q$.
\end{itemize}
\end{lemma}

\begin{proof} W.l.o.g.~we assume $D(d)=1\bmod q$. Notice that $R\ket {G{\setminus}d}\bra {G{\setminus} d} R^\dagger  $ $= \ket {G{\setminus} d}\bra {G{\setminus} d}$ with $R = \prod_{u\in V{\setminus} (B{\cup}\{d\})}{(X_uZ_{\Gamma[V{\setminus} \{d\},V{\setminus} \{d\}].\{u\}})}^{D(u)}$. Moreover $R.Z_{\Gamma.\{u\}}$ is only acting on $V{\setminus}(B{\cup}\{d\})$, so the reduced density matrix for the players in $B$ is \\\centerline{$\begin{array}{cl}&Tr_{V{{\setminus}} (B{\cup} \{d\})}(Z^s_{\Gamma.\{d\}}\ket{G{{\setminus}} d}\bra{G{{\setminus}} d}{Z^{\dagger}}^s_{\Gamma.\{d\}}) \\=  &Tr_{V{{\setminus}} (B{\cup} \{d\})}({(Z_{\Gamma.\{d\}}R)}^s\ket{G{{\setminus}} d}\bra{G{{\setminus}} d}{(Z_{\Gamma.\{d\}}R)}^{\dagger s})\\=& Tr_{V{{\setminus}} (B{\cup} \{d\})}(\ket{G{{\setminus}} d}\bra{G{{\setminus}} d})\end{array}$} which does not depend on the secret, so the players in $B$ have no information about the secret. 

\end{proof}

\begin{corollary}\label{cor:suff_no}
Given a $q$-multigraph $G = (V,\Gamma)$, $d\in V$, and $B\subseteq V{{\setminus}} \{d\}$, if $\pi_G(B,d) =0$ then $B$ has no information about the classical secret for the corresponding CQ encoding.
\end{corollary}

\begin{proof}Let $F = V{{\setminus}} (B{\cup} \{d\})$. According to lemma \ref{lem:suff_no}, $B$ has no information about classical secret if there exists $D:V{{\setminus}} B\rightarrow \F_q$ such that $D(d)=1\bmod q$ and $\Gamma[V{{\setminus}} B,B].D = 0$, 
so if $ \Gamma[F,B].C = -\Gamma[V,B]\{d\}$, where $C:F\to \mathbb F_q= u\mapsto D(u)$ is the restriction of $D$ to $F$. As a consequence, $B$ has no information about classical secret if the system $ \Gamma[F,B].x = -\Gamma[V,B]\{d\}$ has a non zero solution, which is equivalent to find a non zero solution to the system $ \Gamma[F,B].x = \Gamma[V,B]\{d\}$, so if $\rk( \Gamma[F,B]) =  \rk(\Gamma[V{{\setminus}} B,B])$ i.e., $\pi_G(B,d)=0$.
\end{proof}

\noindent {\bf Proof of Theorem \ref{thm:class_char}}. The proof of Theorem \ref{thm:class_char} follows from Corollaries \ref{cor:suff_acc} and \ref{cor:suff_no} and the fact that for every $B$, $0\le \pi_G(B,d)\le 1$. It proves that the encoding is perfect i.e., every set of players is either able to reconstruct the secret (when $\pi_G(B,d)=1$) or has no information about the secret (when $\pi_G(B,d)=0$).

\subsection{Quantum Information}

In the following we prove that the accessibility of a set a players is characterised by the derivative of the cut-rank function with respect to the dealer. \\

\begin{theorem}\label{thm:q_acc}
Given a $q$-multigraph $G$ with a distinguished dealer $d\in V(G)$, a set $B\subseteq V(G){{\setminus}} \{d\}$ of players can recover a quantum secret in the corresponding QQ encoding iff  $$\partial_d\cutrk_G(B)=-1$$ where $\partial_d\cutrk_G(B) = \cutrk_G(B{\cup} \{d\}) - \cutrk_G(B)$ is the discrete derivative of $\cutrk_G$ in $B$ with respect to $d$. 
\end{theorem}

\begin{proof}
It is known that $B$ can access a quantum secret in $G$ iff $B$ can access a classical secret in two mutual unbiased bases, say in $G$ and $G*^1d$ \cite{MM12}. Moreover $B$ can access a classical secret in $G$ iff $\pi_{G}(B,d)=1$, where 
$\pi_{G}(B,d) = \cutrk_{G}(B) -\prk{G}(B,V{\setminus} (B{\cup} \{d\}))$. \\
$(\Rightarrow)$ If $B$ can access a quantum secret, $B$ can access a classical secret and $V{{\setminus}} (B{\cup} \{d\})$ has no information about a quantum secret  \cite{CGL99}, which implies that $V{{\setminus}} (B{\cup} \{d\})$ cannot access  a classical secret. Thus $\pi_G(B,d) = 1$ and $\pi_G(V{{\setminus}} (B{\cup} \{d\}),\{d\})=0$. As a consequence $\pi_G(B,d)-\pi_G(V{{\setminus}} (B{\cup} \{d\}),\{d\}) =1$, so
$1 = \cutrk(B) - \prk{G}(B,V{{\setminus}} (B{\cup} \{d\})) - \cutrk(V{{\setminus}} (B{\cup} \{d\})) +  \prk{G}(V{{\setminus}} (B{\cup} \{d\}),B)= \cutrk(B) -\cutrk(V{{\setminus}} (B{\cup} \{d\})) = \cutrk(B) -\cutrk(B{\cup} \{d\})$.\\
$(\Leftarrow)$ If $\cutrk_G(B)= \cutrk_G(B{\cup} \{d\}) +1$, then $\pi_G(B,\{d\}) = 1$, so $B$ can access a classical secret in $G$. Moreover, since the cut rank is invariant by local complementation \cite{KanteRao}, $\cutrk_{G{\star}^1d}(B)= \cutrk_{G{\star}^1d}(B{\cup} \{d\}) +1$, so $B$ can also access a classical secret in $G{\star}^1d$.
\end{proof}

Notice that for any set $B$ of players, $\partial_d\cutrk_G(B)\in \{-1,0,1\}$: if $\partial_d\cutrk_G(B)=-1$, $B$ can recover the quantum secret; if  $\partial_d\cutrk_G(B)=1$ they have no information since $V{{\setminus}} (B{\cup} \{d\})$ can recover the quantum secret; and if $\partial_d\cutrk_G(B)=0$ they have some partial information about the secret.

Since the cut rank function is submodular \cite{OumSeymour}, 
its derivative is monotonic (decreasing): if $B\subseteq B'$, $\partial_d\cutrk_G(B)\ge \partial_d\cutrk(B')$. Indeed, if $B$ can recover the secret, any superset $B'$ of $B$ can recover it too; and if $B'$ has no information about the secret, any subset $B$ of $B'$  has no information too.

\section{Application to CQ and QQ protocols} \label{Section: secret sharing}

We now see how the encoding of section \ref{Section: encoding}, and the results on access structures in section \ref{Section: access structure} can be used to provide secret sharing protocols.  Following the prescription of \cite{MM12} (based on \cite{MS08,Keet10}, see also \cite{MM13}) we will now introduce two protocols, one for sharing  classical secrets over a quantum channel (CQ) and one for sharing a quantum secret (QQ), both based on a graph state associated with a multigraph. Both protocols can be understood as using the graph state as a channel between the dealer (associated with vertex $d$) and the players (the remaining vertices). In the CQ case this channel is used to perform an Ekert-like key distribution protocol between the dealer and authorised players, so that when completed the dealer and authorised players will share a random `dit' string which is unknown to anybody else. In the QQ case the channel is used to teleport the secret to the players such that only authorised sets of players can access the information (the QQ encoding in section \ref{Section: encoding} can be understood as this teleportation, see Appendix \ref{Section:appendix}). More details on the protocols and their relation to each other as well as error correction can be found in \cite{MM12}.

\subsection{Detailed protocols} \label{section:protocols}

Before we write the full protocols out, we first review some background on the graph state formalism, which will be the key in seeing how the stabilisers can be used to specify how authorised sets can access the information, given the satisfaction of the conditions outlined in the previous section.

Given a multigraph $G=(V,\Gamma)$, we begin with an illustrative expansion of the graph state $|G\rangle_V$ according to the $d$, $V{\setminus} \{d\}$ partition.
\begin{align}
\ket G = \frac1{\sqrt {q^n}} \sum_{x=(x_1,\cdots,x_n)\in \F^n_q}\omega^{|G[x]|}\ket x_V &=  \frac 1{\sqrt q}\sum_i \ket i_d Z^i_{\Gamma.\{d\}}\ket {G{\setminus} d}_{V{\setminus} \{d\}} \nonumber\\
&= \frac 1{\sqrt q}\sum_i \ket i_d  \ket{i_L}_{V{\setminus} \{d\}} \nonumber\\
&= \frac 1{\sqrt q}\sum_i \ket{i(t)}_d  \ket{i'_L(t)}_{V{\setminus} \{d\}}, \nonumber
\end{align}
for any $t\in \F_q$, where the second line follows from definitions in section \ref{Section: encoding}, corresponding to the CQ encoding achieved by the dealer measuring in the $Z$ basis. The third line corresponds to when the dealer measures in bases $X^tZ$ (explained in more detail later), where they are defined as $\ket{i(0)}=\ket{i}$, and
$\ket{i(t)}=\frac 1{\sqrt q} \sum_{j=0}^{q-1} \omega^{\frac{j(j-t)}{2t}-it^{-1}j} \ket j$ for $t=1...q-1$, so that $X^tZ\ket{i(t)}=\omega^i\ket{i(t)}$, and further $\ket{i'(0)_L}=\ket{i_L}=Z^i_{\Gamma.\{d\}}\ket {G{\setminus } d}$ and $\ket{i'_L(t)}:= \frac 1{\sqrt q}\sum_{k=0}^{q-1} \omega^{\frac{-k(k-t)}{2t}+it^{-1}k} \ket{k_L}$ for $t=1...q-1$. The state 
$\ket{i'(t)_L}_{V{\setminus} \{d\}}$ is equivalent to the CQ encoding of $i$ on graph $G*^td$ \cite{Keet10}.

We now look at how the conditions for access arrived at in section \ref{Section: access structure} can be used, along with the stabiliser (or ``fixed point'') condition (\ref{Eqn: stab}), to eventually see how authorised sets can access the information in the CQ and QQ protocols.
We start with the QQ case, which is enough to imply the CQ case (see \cite{MM12}). Suppose a set of players $B\subset V{\setminus} \{d\}$ has access to quantum information in a graph $G=(V, \Gamma)$. We proved with Theorem \ref{thm:q_acc} that $B$ can access QQ encoded quantum information in $G$ if and only if $B$ can access the CQ encoded classical information in $G$ and $V{\setminus} (B{\cup} \{d\})$ cannot. Thus, by rewriting lemma \ref{lem:suff_acc} and \ref{lem:suff_no} applied to $B$ and $V{\setminus} (B{\cup} \{d\})$, we have: there exists $D : B\rightarrow \F_q$ and $C: B{\cup} \{d\}\rightarrow \F_q$ such that $C(d)=1$ and
\begin{IEEEeqnarray}{rCl}
&&\supp (\Gamma[B,V{\setminus} B].D)=\{d\}\label{suff_acc}\\
&&\Gamma[B{\cup} \{d\},V{\setminus} (B{\cup} \{d\})].C=0
\end{IEEEeqnarray}

Now, call $K_i=X_iZ_{\Gamma.\{i\}}$ and $k_i=X_iZ_{\Gamma[V{\setminus} \{d\},V{\setminus} \{d\}]\{i\}}$ (these are the fixpoint operators, or stabilisers for graphs $G$ and $G{\setminus} d$ respectively according to (\ref{Eqn: stab})).\\
First we have $K_C=K_d\prod_{i\in B}K^{C(i)}_i=X_dZ^{\beta}_d. Z_{\Gamma.\{d\}}\prod_{i\in B}k^{C(i)}_i$ with $\beta=\Gamma.C(d)$.
Then $Z_{\Gamma.\{d\}}\prod_{i\in B}k^{C(i)}_i=\omega^{\lambda}\prod_{i\in B}X^{C(i)}_iZ^{\Gamma.C(i)}_i$, with
$\lambda=\sum_{i,j\in B{\cup} \{d\}, j<i}\Gamma(j,i)C(j)C(i)$. \\
Next $K_D$ satisfies $K_D=\prod_{i\in B}K^{D(i)}_i=Z^{\alpha}_d\prod_{i\in B}k^{D(i)}_i$, with $\alpha=\Gamma.D(d)\neq 0$ since (\ref{suff_acc}), and $\prod_{i\in B}k^{D(i)}_i=\omega^{\lambda'}\prod_{i\in B}X^{D(i)}_iZ^{\Gamma.D(i)}_i$, $\lambda'=\sum_{i,j\in B, j<i}\Gamma(j,i)D(j)D(i)$.\\
Later we will suppose $\alpha=1$ (change $D$ to $\alpha^{-1}.D$ if necessary). \\
Hence $\begin{array}[t]{l}{K_C}^t{K_D}^{1-t\beta}\ket G\\=\omega^{\frac{t(t-1)}{2}\beta}X^t_dZ_d.{[Z_{\Gamma.\{d\}}\prod_{i\in B}k^{C(i)}_i]}^t{[\prod_{i\in B}k^{D(i)}_i]}^{1-t\beta} \ket G\\=\ket G\end{array}$ \\
which is a stabiliser / fixpoint equation involving operators only in $d$ and $B$ which will be used to inform which measurements should be made to recover the secret in the CQ case, and how to find the QQ decoding operation.
We can rewrite this as follows \\
${[ Z_{\Gamma.\{d\}}\prod_{i\in B}k^{C(i)}_i]}^t{[\prod_{i\in B}k^{D(i)}_i]}^{1-t\beta}=\omega^c\prod_{i\in B}X^{x_i}_iZ^{z_i}_i$
with
\begin{align}
x_i(t)&=tC(i)+(1-t\beta) D(i) \label{Eqn: x}\\
z_i(t)&=t \Gamma.C(i)  +(1-t\beta)\Gamma.D(i),\label{Eqn: z}\\
c=t\lambda'&+(1-t\beta)\lambda+t(t-1)\lambda'+(1-t\beta)(-t\beta)\lambda+t(1-t\beta)\sum_{i,j\in B}\Gamma(i,j)C(i)D(j)
\end{align}
and we further define
\begin{align} \label{Eqn: f}
f_t(r):=-r-c-\frac{t(t-1)}{2}\beta.
\end{align}
We can then see that given the state $|G\rangle_V$, if the dealer measures $X^tZ$, getting result $\omega^{s(t)}$ and each player $i$ in $B$ measures its qudit in the $X^{x_i(t)}Z^{z_i(t)}$ bases, denoting their results $m_i(t)$, if we define $m(t)=f^{-1}_t(\sum_i m_i(t))$, then the fixpoint stabiliser conditions imply $m(t)=s(t)$. This will be the basis of the CQ accessing strategy.

For the QQ accessing, we define operators $U_B$ and $V_B$ only acting on $B$ such that
$U_B:=\prod_{i\in B}k^{-D(i)\alpha^{-1}}_i$, which satisfies $U_B\ket{s_L}=\omega^{s}\ket{s_L}$
and $V_B:=Z_{\Gamma.\{d\}}\prod_{i\in B}k^{C(i)-\beta\alpha^{-1}D(i)}_i$, which satisfies $V_B\ket{s_L}=\ket{(s+1)_L}$.\\
We also define the extended Bell basis as the following bipartite states
over a system $\{a,b\}$: $\forall k,l\in F_q$,
$\ket{\beta_{k,l}}_{ab}=Z_a^kX_b^l\sum_{i\in\F_q}\frac{\ket{ii}_{ab}}{\sqrt{q}}$.
The result $(k,l)$ of a measurement over $\{a,b\}$ in the Bell basis
yield the state as $\ket{\beta_{k,l}}_{ab}$.

\begin{description}
\item
\textbf{CQ Protocol:}\ \\
\begin{enumerate}
\item The dealer prepares the graph state
\begin{equation*}
\ket{G}=\sum_{i=0}^{q-1}\frac 1{\sqrt{q}}\ket{i(t)\,}_d\ket{i(t)_L}_{V{\setminus} \{d\}}
\end{equation*}
and sends one qudit of the state to each player.
\item The dealer randomly measures its qudit among the bases: $\{X^tZ\}_{t\in T}$ and denotes the result $\omega^{s(t)}$. That leaves the state over the players on $\ket{i(t)_L}_{V{\setminus} \{d\}}$.
\item A player $b\in B$  randomly chooses $t'\in T$ and send $t'$ to the other players in $B$ using their private channel.
\item Each player $i$ in $B$ measures its qudit in the $X^{x_i(t')}Z^{z_i(t')}$ bases (see (\ref{Eqn: x}),(\ref{Eqn: z})) and sends the result $\omega^{m_i(t')}\in \{1,\omega, .., \omega^{q-1}\}$ to $b$.
\item $b$ computes $m(t')=f_{t'}^{-1}(\sum_im_{i}(t'))$ (see (\ref{Eqn: f})).
\item Repeat step $1.$ $2.$ $3.$ $p\rightarrow\infty$ times. The list of measurement results $s(t)$ and $m(t')$ are the raw keys of the dealer and players $B$ respectively.
\item \textsc{security test}: Follow standard QKD security steps. Through public discussion between $d$ and $B$ first sift the key followed by standard error correction and privacy amplification to generate a secure key (see \cite{MM12} and \cite{SS10}).
\end{enumerate}
\end{description}

Correctness : After the QKD security steps the dealer and the authorised set $B$ will be able to share a perfectly secure random key. Furthermore, QQ unauthorised sets for the same graph will not be able to establish such a key (see \cite{MM12} for proofs).\\

\begin{description}
\item
\textbf{QQ Protocol:}
Let $|\zeta\rangle_{S}= \sum_{i=0}^{q-1}s_i\ket{i}_S \in \C^q$ be a quantum secret.\\
\begin{enumerate}
\item A dealer prepares the state 
\begin{equation*}\frac{1}{\sqrt{q}}\sum_{0\leq i\leq q-1} s_i Z^i_{\Gamma.\{d\}}\ket{G{\setminus} d}_{V{\setminus} \{d\}}\end{equation*}
\item The dealer sends one qudit of the resultant state to each player.
\item (measurement) The authorized set $B$ uses two ancillas qudits $\{a_1,a_2\}$ prepared in the Bell pair state $\ket{\beta_{00}}_{a_1a_2}$, and performs the following two commuting projective measurement on $a_1$, $V^{-1}_BX^{-1}_{a_1}$ and $U_BZ^{-1}_{a_1}$ on ,  with result denoted $k$ and $l$ respectively.
\item (correction) $B$ applies $Z^{k}X^{-l}$ over the second ancilla $\{a_2\}$.
\end{enumerate}
\end{description}

Correctness: $U_B$ and $V_B$ satisfy $U_B\ket{i_L}_{V{\setminus} \{d\}}=\omega^i\ket{i_L}_{V{\setminus} \{d\}}$, and $V_B\ket{i_L}_{V{\setminus} \{d\}}=\ket{(i+1)_L}_{V{\setminus} \{d\}}$ $\forall i\in\F_q$. We can rewrite the state over $V{\setminus} \{d\}{\cup} \{a_1, a_2\}$ as:\\
$\begin{array}{ll}
&\sum_{i\in \F_q }s_i\ket{i_L}_{V{\setminus} \{d\}}\sum_{j\in\F_q}\frac{\ket{jj}_{a_1 a_2}}{\sqrt{q}}.\\
=&\frac{1}{\sqrt{q}}\sum_{l\in \F_q}I_{V{\setminus} \{d\}}X_{a_1}^lX_{a_2}^l\sum_{i\in\F_q}\ket{i_Li}_{V{\setminus} \{d\}a_1}s_i\ket{i}_{a_2}\\
=&\frac{1}{\sqrt{q}}\sum_{l\in \F_q}I_{V{\setminus} \{d\}}X_{a_1}^lX_{a_2}^l\sum_{k\in\F_q}\sum_{i\in\F_q}\omega^{k.i}\frac{\ket{i_Li}_{{V{\setminus} \{d\}}a_1}}{q}\sum_j\omega^{-k.j}s_j\ket{j}_{a_2}\\
=&\frac{1}{\sqrt{q}}\sum_{l\in \F_q}I_{V{\setminus} \{d\}}X_{a_1}^lX_{a_2}^l\sum_{k\in\F_q}U^k_BI_{a_1}Z_{a_2}^{-k}\sum_{i\in\F_q}\frac{\ket{i_Li}_{{V{\setminus} \{d\}}a_1}}{q}\sum_js_{j}\ket{j}_{a_2}\\
=&\frac{1}{q}\sum_{l,k\in \F_q}U^k_BX_{a_1}^l\sum_{i\in\F_q}\frac{\ket{i_Li}_{{V{\setminus} \{d\}}a_1}}{\sqrt{q}}X_{a_2}^lZ_{a_2}^{-k}\sum_{j\in\F_q}s_j\ket{j}_{a_2})
\end{array}$\\

As $V^{-1}_BX^{-1}_{a_1}(U^k_BX_{a_1}^l\sum_{i\in\F_q}\frac{\ket{i_Li}_{{V{\setminus} \{d\}}a_1}}{\sqrt{q}})=\omega^kU^k_BX_{a_1}^l\sum_{i\in\F_q}\frac{\ket{i_Li}_{V{\setminus} \{d\}a_1}}{\sqrt{q}}$ and \\
$U_BZ^{-1}_{a_1}(U^k_BX_{a_1}^l\sum_{i\in\F_q}\frac{\ket{i_Li}_{{V{\setminus} \{d\}}a_1}}{\sqrt{q}})=\omega^lU^k_BX_{a_1}^l\sum_{i\in\F_q}\frac{\ket{i_Li}_{{V{\setminus} \{d\}}a_1}}{\sqrt{q}}$, the projective measurement according to $V^{-1}_BX^{-1}_{a_1}$ and $U_BZ^{-1}_{a_1}$ reveals the syndrome $(k,l)$, such that the correction  $Z^{k}X^{-l}$ over the ancilla $\{a_2\}$ leaves the state as $\sum_i s_i\ket{i}_{a_2}$.

\subsection{Example}
We illustrate the use of characterisation of the access structure in a multigraph with a Reed Solomon Graph State that allows a quantum secret (or equivalently a random key of $dits$) to be shared between a dealer and all subset of at least $\frac{n+1}{2}$ players among a set of $n$ players over a field of $q$ elements, with $q\geq n$. We refer to \cite{MM13}, \cite{CGL99} for more details about Reed Solomon Graph for secret sharing. \\
We saw $B\subset V{\setminus} \{d\}$ can access quantum information with respect to $d$ in $G$ iff  there exist $D : B\rightarrow \F_q$ and $C: B{\cup} \{d\}\rightarrow \F_q$ such that $C(d)=1$ and
\begin{IEEEeqnarray*}{rCl}
&&\supp (\Gamma[B,V{\setminus} B].D)=\{d\}\\
&&\Gamma[B{\cup} \{d\},V{\setminus} (B{\cup} \{d\})].C=0
\end{IEEEeqnarray*}
We rewritte these conditions in the following way: $B\subset V{\setminus} \{d\}$ can access quantum information in $G$ iff there exist $D : B\rightarrow \F_q$ and $C: B \cup \{d\}\rightarrow \F_q$ such that $C(d)=1$ and $\left\{\begin{array}{ll}&B\ \cup\ \{u\in V {\setminus} B \quad | \quad \sum_{v\in B^{}}D(v).\Gamma(u,v)\neq 0\mod q\}=B \cup \{d\}. \ \ \hfill (5)\\
&B\ \cup\ \{d\} \cup\  \{u\in V{\setminus} (B{\cup} \{d\}) \ | \ \sum_{v\in B{\cup} \{d\}}C(v).\Gamma(v,u)\neq 0\mod q\}=B \cup \{d\} \ \ \hfill (6)\end{array}\right.$

For $A: V \rightarrow \F_q$, we call $G_A=(V_A,\Gamma_A)$ the subgraph induced by $A$ such that:\\
$V_A=\supp(A) \cup \{v\in V{\setminus} \supp(A)\quad |\quad  \Gamma[\supp(A),V{\setminus} \supp(A)].A(v)\neq 0\mod q$\\
and 
$\forall v_i\in \supp(A), \left\{\begin{array}{l}\Gamma_A(v_i,v_j)=A(v_i)A(v_j).\Gamma(v_i,v_j)\text{ if } v_j\in \supp(A)\\
    \Gamma_A(v_i,v_j)=A(v_i)\Gamma(v_i,v_j)\text{ if } v_j\in V_A{\setminus} \supp(A)\end{array}\right.$ \\

For example, let $G=(V,\Gamma)$, $d\in V$, $|V|=8$, be the $(4,3,7)_7$ Reed Solomon Graph State given in Fig \ref{RS74}.

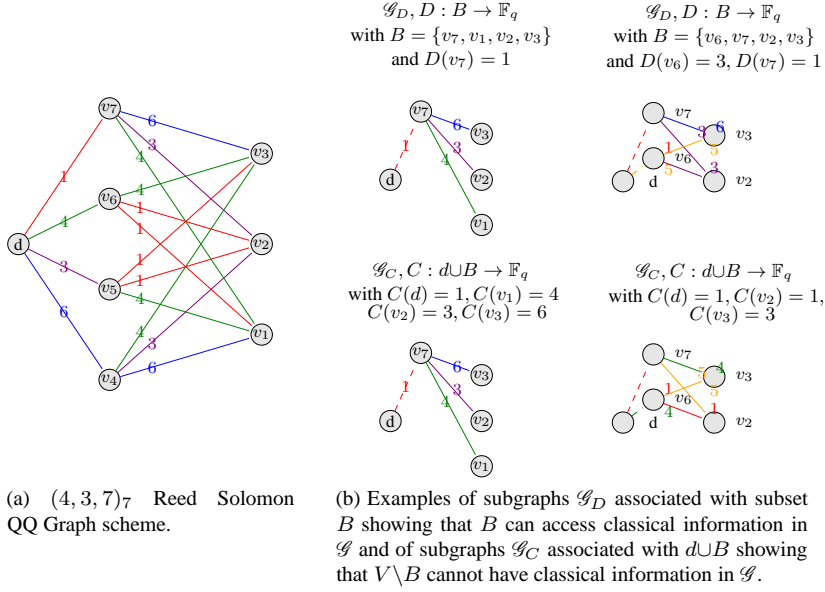
\begin{figure}[t]
\vspace{-1cm}\subfloat[1][$(4,3,7)_7$ Reed Solomon $\quad$ QQ~Graph scheme.]{%
\scalebox{0.8}{
\begin{tikzpicture}[shorten >=1, -, font=\footnotesize]
\tikzstyle{vertex}=[circle,draw,fill=black!10,minimum size=10pt,inner sep=0pt,font=\footnotesize]

\node[vertex] (d) at (0,2.25) {d};
\foreach \name/\y/\text in {P-1/0/v_4, P-2/1.5/v_5, P-3/3/v_6, P-4/4.5/v_7}
  \node[vertex,xshift=0cm,yshift=0cm] (\name) at (1.5,\y) {$\text$};

\foreach \name/\y/\text in {Q-1/0.75/v_1, Q-2/2.25/v_2, Q-3/3.75/v_3}
  \node[vertex,xshift=0cm,yshift=0cm] (\name) at (4,\y) {$\text$};

\foreach \to in {1}
           { \draw[blue] (d) -- node {$6$}(P-\to);}
\draw[purple] (d) -- node{$3$} (P-2);
\draw[green] (d) -- node{$4$} (P-3);
\draw[red] (d) -- node{$1$} (P-4);
\foreach \from/\to in {4/3}
           { \draw[blue] (P-\from) node[xshift=0.7cm,yshift=-0.2cm] {$6$} -- (Q-\to);}
\foreach \from/\to in {1/1}
           { \draw[blue] (P-\from) node[xshift=0.7cm,yshift=0.2cm] {$6$} -- (Q-\to);}
\foreach \from/\to in {4/2}
           { \draw[purple] (P-\from) node[xshift=0.7cm,yshift=-0.6cm] {$3$} -- (Q-\to);}
\foreach \from/\to in {1/2}
           { \draw[purple] (P-\from) node[xshift=0.7cm,yshift=0.6cm] {$3$} -- (Q-\to);}
\foreach \from/\to in {1/3}
           { \draw[green] (P-\from)  node[xshift=0.5cm,yshift=0.8cm] {$4$} -- (Q-\to);}
\foreach \from/\to in {4/1}
           { \draw[green] (P-\from)  node[xshift=0.5cm,yshift=-0.8cm] {$4$} -- (Q-\to);}
\foreach \from/\to in {2/1}
           { \draw[green] (P-\from)  node[xshift=0.5cm,yshift=-0.15cm] {$4$} -- (Q-\to);}
\foreach \from/\to in {3/3}
           { \draw[green] (P-\from)  node[xshift=0.5cm,yshift=0.15cm] {$4$} -- (Q-\to);}

\foreach \from/\to in {2/2}
           { \draw[red] (P-\from)  node[xshift=0.5cm,yshift=0.15cm] {$1$} -- (Q-\to);}
\foreach \from/\to in {3/2}
           { \draw[red] (P-\from)  node[xshift=0.5cm,yshift=-0.15cm] {$1$} -- (Q-\to);}

\foreach \from/\to in {2/3}
           { \draw[red] (P-\from)  node[xshift=0.5cm,yshift=0.5cm] {$1$} -- (Q-\to);}
\foreach \from/\to in {3/1}
           { \draw[red] (P-\from)  node[xshift=0.5cm,yshift=-0.5cm] {$1$} -- (Q-\to);}


\node (rien) at (0,-1.5) {$ $};
\node (rien) at (5,2.2) {$ $};
\end{tikzpicture}}
\label{RS74}}
\subfloat[1][Examples of subgraphs $\mathscr{G}_D$ associated with subset $B$ showing that $B$ can access classical information in $\mathscr{G}$ and of subgraphs $\mathscr{G}_C$ associated with $d{\cup} B$ showing that $V{\setminus} B$ cannot have classical information in $\mathscr{G}$.]{%
\scalebox{0.8}{
\begin{tikzpicture}[shorten >=1, -, font=\footnotesize]

\tikzstyle{vertex}=[circle,draw,fill=black!10,minimum size=10pt,inner sep=0pt,font=\footnotesize]
\tikzstyle{texte}=[font=\small]
\node[texte]  at (1.5,3.1) {$\text{and }D(v_7)=1$};
\node[texte]  at (1.5,3.5) {$\text{with }B=\{v_7,v_1,v_2,v_3\}$};
\node[texte]  at (1.5,3.9) {$\mathscr{G}_D, D:B\rightarrow \F_q$};

\node[vertex] (d) at (0.5,1.12) {d};
\foreach \name/\y/\text in { P-4/2.25/v_7}
  \node[vertex,xshift=0cm,yshift=0cm] (\name) at (1,\y) {$\text$};

\foreach \name/\y/\text in {Q-1/0.37/v_1, Q-2/1.12/v_2, Q-3/1.88/v_3}
  \node[vertex,xshift=0cm,yshift=0cm] (\name) at (2,\y) {$\text$};

\draw[red, dashed] (d) -- node{$1$} (P-4);
\foreach \from/\to in {4/3}
           { \draw[blue] (P-\from) node[xshift=0.6cm,yshift=-0.22cm] {$6$} -- (Q-\to);}
\foreach \from/\to in {4/2}
           { \draw[purple] (P-\from) node[xshift=0.6cm,yshift=-0.6cm] {$3$} -- (Q-\to);}
\foreach \from/\to in {4/1}
           { \draw[green] (P-\from)  node[xshift=0.4cm,yshift=-0.8cm] {$4$} -- (Q-\to);}

\node[texte]  at (1.5,-0.4) {$\mathscr{G}_C, C:d{\cup} B\rightarrow \F_q$};
\node[texte] at (1.5,-0.8) {$\text{ with }C(d)=1,C(v_1)=4$} ;
\node[texte] at (1.6,-1.1) {$C(v_2)=3,C(v_3)=6$} ;

\node[vertex] (d2) at (0.5,-2.88) {d};
\foreach \name/\y/\text in { P-42/-1.75/v_7}
  \node[vertex,xshift=0cm,yshift=0cm] (\name) at (1,\y) {$\text$};

\foreach \name/\y/\text in {Q-12/-3.63/v_1, Q-22/-2.88/v_2, Q-32/-2.12/v_3}
  \node[vertex,xshift=0cm,yshift=0cm] (\name) at (2,\y) {$\text$};

\draw[red, dashed] (d2) -- node{$1$} (P-42);
\foreach \from/\to in {42/32}
           { \draw[blue] (P-\from) node[xshift=0.6cm,yshift=-0.22cm] {$6$} -- (Q-\to);}
\foreach \from/\to in {42/22}
           { \draw[purple] (P-\from) node[xshift=0.6cm,yshift=-0.6cm] {$3$} -- (Q-\to);}
\foreach \from/\to in {42/12}
           { \draw[green] (P-\from)  node[xshift=0.4cm,yshift=-0.8cm] {$4$} -- (Q-\to);}

\end{tikzpicture}
\begin{tikzpicture}[shorten >=1, -, font=\footnotesize]
\hspace{0.5cm}
\tikzstyle{vertex}=[circle,draw,fill=black!10,minimum size=10pt,inner sep=0pt,font=\footnotesize]
\tikzstyle{texte}=[font=\small]
\node[texte] at (1.5,3.1) {$\text{and } D(v_6)=3, D(v_7)=1$} ;
\node[texte] at (1.5,3.5) {$\text{with }B=\{v_6,v_7,v_2,v_3\}$};
\node[texte]  at (1.5,3.9) {$\mathscr{G}_D, D:B\rightarrow \F_q$};

\node[texte] at (2,0.37) {$ $};

\node[vertex] (d) at (0.5,1.12) {d};
\foreach \name/\y/\text in {P-3/1.5/v_6, P-4/2.25/v_7}
  \node[vertex,xshift=0cm,yshift=0cm] (\name) at (1,\y) {$\text$};

\foreach \name/\y/\text in {Q-2/1.12/v_2, Q-3/1.88/v_3}
  \node[vertex,xshift=0cm,yshift=0cm] (\name) at (2,\y) {$\text$};

\draw[orange, dashed] (d) -- node{$5$} (P-3);
\draw[red, dashed] (d) -- node{$1$} (P-4);
\foreach \from/\to in {4/3}
           { \draw[blue] (P-\from) node[xshift=0.6cm,yshift=-0.22cm] {$6$} -- (Q-\to);}
\foreach \from/\to in {4/2}
           { \draw[purple] (P-\from) node[xshift=0.3cm,yshift=-0.3cm] {$3$} -- (Q-\to);}
\foreach \from/\to in {3/3}
           { \draw[orange] (P-\from)  node[xshift=0.5cm,yshift=0.15cm] {$5$} -- (Q-\to);}
\foreach \from/\to in {3/2}
           { \draw[purple] (P-\from)  node[xshift=0.5cm,yshift=-0.15cm] {$3$} -- (Q-\to);}

\node[texte]  at (1.5,-0.4) {$\mathscr{G}_C, C:d{\cup} B\rightarrow \F_q$};
\node[texte] at (1.5,-0.8) {$\text{ with }C(d)=1,C(v_2)=1,$} ;
\node[texte] at (1.8,-1.1) {$ C(v_3)=3$} ;

\node[vertex] (d2) at (0.5,-2.88) {d};
\foreach \name/\y/\text in {P-32/-2.5/v_6, P-42/-1.75/v_7}
  \node[vertex,xshift=0cm,yshift=0cm] (\name) at (1,\y) {$\text$};

\foreach \name/\y/\text in {Q-22/-2.88/v_2, Q-33/-2.12/v_3}
  \node[vertex,xshift=0cm,yshift=0cm] (\name) at (2,\y) {$\text$};

\draw[green, dashed] (d2) -- node{$4$} (P-32);
\draw[red, dashed] (d2) -- node{$1$} (P-42);
\foreach \from/\to in {42/32}
           { \draw[green] (P-\from) node[xshift=0.6cm,yshift=-0.22cm] {$4$} -- (Q-\to);}
\foreach \from/\to in {42/22}
           { \draw[orange] (P-\from) node[xshift=0.3cm,yshift=-0.3cm] {$5$} -- (Q-\to);}
\foreach \from/\to in {32/32}
           { \draw[orange] (P-\from)  node[xshift=0.5cm,yshift=0.15cm] {$5$} -- (Q-\to);}
\foreach \from/\to in {32/22}
           { \draw[red] (P-\from)  node[xshift=0.5cm,yshift=-0.15cm] {$1$} -- (Q-\to);}

\node[texte] at (2,-3.67) {$ $};
\end{tikzpicture}}
\label{SBp}}
\caption{Checking quantum accessibility in a $(4,3,7)_7$ Reed Solomon Graph.}
\end{figure}

\begin{table}[t]
\center\scalebox{0.8}{\footnotesize{
\begin{tabular}{|c|c|c|c|c|c|}
\hline
$B$ & $(D(b))_{b\in Bs}$ & $(C(b))_{b\in d{\cup} B}$ & $B$ & $(D(b))_{b\in Bs}$ & $(C(b))_{b\in d{\cup} B}$ \\
\hline
$\{v_7,v_1,v_2,v_3\}$ & $(1,0,0,0)$ & $(1,0,6,0,0)$ &
$\{v_6,v_1,v_2,v_3\}$ & $(1,0,0,0)$ & $(1,0,2,2,1)$\\
$\{v_5,v_1,v_2,v_3\}$ & $(1,0,0,0)$ & $(1,0,3,4,1)$ &
$\{v_4,v_1,v_2,v_3\}$ & $(1,0,0,0)$ & $(1,0,4,6,2)$\\
$\{v_6,v_7,v_2,v_3\}$ & $(3,1,0,0)$ & $(1,0,0,1,3)$ & 
$\{v_6,v_7,v_1,v_2\}$ & $(1,4,0,0)$ & $(1,0,0,3,6)$\\ 
$\{v_6,v_7,v_1,v_3\}$ & $(4,1,0,0)$ & $(1,0,0,5,5)$ & 
$\{v_5,v_7,v_2,v_3\}$ & $(3,4,0,0)$ & $(1,0,0,6,1)$\\ 
$\{v_5,v_7,v_1,v_2\}$ & $(1,1,0,0)$ & $(1,0,0,2,1)$ & 
$\{v_5,v_7,v_1,v_3\}$ & $(4,3,0,0)$ & $(1,0,0,1,4)$\\ 
$\{v_4,v_7,v_2,v_3\}$ & $(4,1,0,0)$ & $(1,0,0,2,2)$ & 
$\{v_4,v_7,v_1,v_3\}$ & $(3,4,0,0)$ & $(1,0,0,6,1)$\\ 
$\{v_5,v_6,v_1,v_2\}$ & $(3,1,0,0)$ & $(1,0,0,1,3)$ &
$\{v_5,v_6,v_1,v_3\}$ & $(1,6,0,0)$ & $(1,0,0,4,3)$\\ 
$\{v_5,v_6,v_7,v_3\}$ & $(2,2,1,0)$ & $(1,0,0,0,2)$ &
$\{v_5,v_6,v_7,v_2\}$ & $(4,1,1,0)$ & $(1,0,0,0,5)$\\ 
$\{v_5,v_6,v_7,v_1\}$ & $(5,6,1,0)$ & $(1,0,0,0,6)$ &
$\{v_4,v_5,v_7,v_3\}$ & $(1,1,1,0)$ & $(1,0,0,0,6)$\\ 
$\{v_4,v_5,v_7,v_1\}$ & $(4,6,1,0)$ & $(1,0,0,0,3)$ & 
$\{v_4,v_5,v_7,v_2\}$ & $(6,1,4,0)$ & $(1,0,0,0,3)$\\ 
$\{v_4,v_5,v_6,v_7\}$ & $(5,6,1,2)$ & $(1,0,0,0,0)$ & & & \\
\hline
\end{tabular}
}}
\caption{List of typical subsets $B$ of $4$ players in the Reed Solomon Graph State described in Fig \ref{RS74}. For each $B$, $B {\cup} \{ u\in V {\setminus} B \ |\ \sum_{v\in B}D(v).\Gamma(u,v)\neq 0\mod q\}=B{\cup} \{d\}=B{\cup} \{d\} {\cup} \{u\in V{\setminus} (B{\cup} \{d\}) \ |\ \sum_{v\in B{\cup} \{d\}}C(v).\Gamma(v,u)\neq 0\mod q\}$, meaning that $B$ can access quantum information, whereas $V{\setminus} (B{\cup} \{d\})$, that is all subset of $3$ players, cannot. (The remaining subsets are covered by symmetry.)}
\label{SBt}
\end{table}

Such a graph can be used by dealer $d$ to encode any quantum secret $\ket{\xi}\in \C^7$ and share it between $7$ players such that all subset of at least $4$ players can recover the secret, whereas any subset of $3$ players or less cannot have any information about it. We can reprove this result using the previous graph characterisation, that is by checking if conditions (5) (6) are satisfied in a basis $G$ for all subset $B\subset V{\setminus} \{d\}$ of $4$ players. In fig \ref{SBp}, we give the relevant induced subgraphs for thee different subset $B$. And in table \ref{SBt} we give a list of relevant multi subset $D : B \rightarrow \F_q$ and $C : B{\cup} \{d\} \rightarrow \F_q$  for typical subsets $B$ of four players.

\section{Existence of $((k,n))_q$ scheme}

In this section, we focus on the properties of the secret sharing scheme realised by a given $\F_q$-graph, as well as the existence of $\F_q$-graphs realising a given secret sharing protocol. A $\F_q$-graph $G$ of order $n$ with a particular dealer $d$ is said to realise a 
$((k,n))_q$ scheme if $k-1 = \max_{B\subseteq V{\setminus} \{d\}} (\partial_d\cutrk_G(B)\ge 0)$. In other words,  $G$ realises a $((k,n))_q$ scheme if all sets of at least $k$ players can recover a quantum secret and there exists a set of $k-1$ players which cannot.  
 A $\F_q$-graph  which realises an $((k,n))_q$ scheme can be used as an $(k,k'\geq n-k,n)_q$ CQ protocol or $(k,n-k,n)_q$ QQ protocol as described in section \ref{Section: secret sharing} (note that they can also be used for $(k,k'\geq n-k,n)_q$ schemes to share a quantum secret using hybrid protocols (e.g. \cite{B09,JMP11,FG11,Gheo12})). \\


\subsection{Finding new schemes}\label{sec:problem}
Theorem \ref{thm:q_acc} offers a combinatorial characterisation of quantum accessibility, and raises as a consequence several questions about the complexity of deciding: $(i)$ whether a given set of players can access a quantum secret in a given $q$-multigraph? $(ii)$ whether a given $q$-multigraph realises a $((k,n))$ protocol? $(iii)$ whether, given $q$, $n$ and $k$,  there exists an $\F_q$-graphs realising a $((k,n))$ protocol?

\begin{description}
\item[Problem $i$] Given an $\F_q$-graph $G$ of order $n$ with a particular dealer  $d$ and a set $B$ of $k$ players, deciding whether $B$ can access a quantum secret consists of deciding whether $\partial _d\cutrk_G(B)=-1$. This  can be decided efficiently since  $\partial _d\cutrk_G(B)$ is computed in $O(nk^{1.38})$ operations using the Gaussian elimination for computing the rank \cite{BunchHop,IbarraMH}. 

\item[Problem $ii$] Given a $\F_q$-graph $G$ of order $n$ and $\alpha \in[0,1]$, deciding whether  $G$ is a $((\alpha n,n))$ scheme can be done by enumerating all the ${\alpha n\choose n}$ sets of players of size $\alpha n$ and for each of them deciding whether they can access a quantum secret. It leads to $O(n^{2.38}2^{nH_2(\alpha)})$ operations. This problem is NP-complete, as it has been shown to be NP complete when $q=2$ \cite{CattaneoPerdrix}, and also hard in terms of  parameterised complexity as it is hard for  $W$[1] \cite{CattaneoPerdrix}. 

\item[Problem $iii$] Given $n,\alpha$, and $q$, deciding whether there exists a $((\alpha n,n))$ $\F_q$-graph? A brute-force approach consists in enumerating all the $q^{\frac{n(n-1)}2}$ $\F_q$-graphs of order $n$ and then decide whether they realise a $((\alpha n,n))$ protocol. It leads to $O(q^\frac{n(n-1)}{2}n^{2.38}2^{nH_2(\alpha)})$ operations. This can be implemented for small values of $n$ only and permits to prove that there is no  $(4,3,7)_3$ QQ secret sharing with qutrit graph state.

\end{description}
Solving problem $i$ can be done with the similar algortihm C of \cite{GLG09}. Note that for one thing, the later is more general and can be applied to input states (that is quantum secrets) and to multigraphs of arbitrary dimension (not necessarily prime number). For another thing, it concerns rather the access to partial information. Also it is not optimised for problem $i$ of our particular interest.

In the following sections, we develop a different approach for deciding the existence of $((\alpha n,n))$ $\F_q$-graphs realising.  We show an upper and a lower bound on the minimal $\alpha$ such that there exists an $\F_q$-graph realising a $((\alpha n, n))$ protocol. The upper bound (Theorem \ref{thm:exist}) is based on non constructive probabilistic methods, whereas the the lower bound (Theorem \ref{thm:lower}) is based on a counting argument.

\subsection{Existence of $q$-multigraphs realising $((\alpha n,n))_q$ schemes}

In this section, we prove a Gilbert-Varshamov-like result: for any $\alpha$ such that $H_{q^2}(1-\alpha)<\frac 12$ there exists a $q$-multigraph realising a $((\alpha n,n))_q$ scheme. The proof is using probabilistic methods and is, as a consequence, non constructive. However, we prove that a random $q$-multigraph satisfies such $((\alpha n,n))_q$ scheme with high probability as long as $H_{q^2}(1-\alpha)<\frac 12$.

\begin{lemma}
For any $q$-multigraph $G=(V,\Gamma)$ of order $n$, and any $\alpha\in [0.5,1]$, if for any multiset $ C: V\to \mathbb F_q$, $|\supp(C){\cup} \supp(\Gamma.C)|>  (1-\alpha)n$ then for any $d\in V$ and any $B\subseteq V{{\setminus}} \{d\}$ such that $|B|\ge \alpha n$, $\partial_d\cutrk_G(B) = -1$. 
\label{lemma1}
\end{lemma}

\begin{proof} 

For any $B\subseteq V$ such that $|B|\ge \alpha n$, $\ker (\Gamma[V{{\setminus}} B]) = \{0\}$, otherwise there would be a multiset $C$ such that $\supp(C)\subseteq V{{\setminus}} B$ and $|\supp(C){\cup} \supp(\Gamma.C)|\le (1-\alpha) n$. So for any $B\subseteq V$ such that $|B|\ge \alpha n$, $\cutrk_G(B) = n- |B|$. As a consequence,  for any $d\in V$ and any $B\subseteq V{{\setminus}} \{d\}$ such that $|B|\ge \alpha n$, $\partial_d\cutrk_G(B) =  n-|B{\cup} \{d\}| - (n-|B|) = -1$. Thus $\partial_d\cutrk_G(B) = -1$
\end{proof}

A random $\F_q$-graph $G(n,1/q)$ is a  $\F_q$-graph of order $n$ such that, for every pair of vertices $u$ and $v$,  the number of edges between $u$ and $v$ is chosen uniformly at random in $\F_q$.

\begin{theorem}\label{thm:exist} Given $q\ge 2$, and $\alpha\in [0.5,1]$ such that $H_{q^2}(1-\alpha)<\frac 12$, for any $n\in \mathbb N$, a random $q$-multigraph $G(n,1/q)$ realises a $((\alpha n,n))_q$ scheme with probability  $1-2^{\Omega(n)}$, where $d$ is any vertex of $G(n,1/q)$. 

 \end{theorem}
\begin{proof}
Let $\mathcal C_\alpha = \{C:V\to \mathbb F_q, |\supp(C)|\le (1-\alpha)n\}$. For any $C\in \mathcal C_\alpha$, let $A_C$ be the (bad) event $|\supp(C){\cup} \supp(\Gamma.C)|\le  (1-\alpha)n$. \\For any $C\in \mathcal C_\alpha$, $Pr(A_C)=\frac1{q^{(1-c)n}} \sum_{k=0}^{(1-\alpha-c)n} {(1-c)n\choose k} (q-1)^{k}$ where $c=|\supp(C)|/n$, and $\sum_{C\in \mathcal C_\alpha} Pr(A_C) = \sum_{j=0}^{(1-\alpha)n} f(j)$ with $f(j) = \sum_{C s.t. |\supp(C)| = j} Pr(A_C)$. \\
In the following, we show an upperbound on $f(k)$. For any $c\in [0,0.5]$, $f(cn)= {n\choose cn}(q-1)^{cn}  \frac1{q^{(1-c)n}}\sum_{k=0}^{(1-\alpha-c)n}{(1-c)n\choose k}(q-1)^{k}\le \frac{(q-1)^{cn}}{q^{(1-c)n}}2^{nH_2(c)+(1-c)nH_2(\frac{1-\alpha -c}{1-c})} (q-1)^{(1-\alpha-c)n} = 2^{ng(c)}$ where $g(c) = {H_2(c)+(1-c)H_2(\frac \alpha {1-c})+(1-\alpha)\log_2(q-1) - (1-c)\log_2(q)}$. $g'(c) = -\log_2(c)+\log_2(1-\alpha-c)+\log_2(q)$, so $g'(c)=0 \iff c= \frac q{q+1}(1-\alpha)$. As a consequence, $g(c)\le g( \frac q{q+1}(1-\alpha)) = -\alpha\log_2(\alpha)-(1-\alpha)\log_2(\alpha) +(1-\alpha)\log_2(q^2-1) - \log_2(q) = \log_2(q)(2H_{q^2}(1-\alpha) -1)$. Thus, $\sum_{C\in \mathcal C_\alpha} Pr(A_C) \le (1-\alpha)nq^{n[2H_{q^2}(1-\alpha)-1]}$, so, thanks to the union bound, $Pr(\bigcap_{C\in \mathcal C_\alpha} \overline{A_C})\ge 1- (1-\alpha)nq^{n[2H_{q^2}(1-\alpha)-1]}= 1-2^{\Omega(n)}$ when $2H_{q^2}(1-\alpha)-1<0$. So according to lemma \ref{lemma1}, $\kappa_Q(G,d) \le \alpha n$ for any vertex $d$ when $H_{q^2}(1-\alpha)<\frac 12$.
\end{proof}

Theorem \ref{thm:exist} extends the upper bound of the binary case ($q=2$) \cite{JMP11}. Notice that even if a random $\F_q$-graph realises a $((\alpha n, n))_q$ scheme with probability almost $1$, double checking whether a (randomly chosen) $\F_q$-graph actually realises a $((\alpha n, n))_q$ scheme is a hard task (see Problem ($ii$) in section \ref{sec:problem}).

\subsection{Lower bound on quantum accessibility}

The no cloning theorem 
implies that for any $((\alpha n, n))$ secret sharing protocol, $\alpha > 0.5$. In the following we improve this lower bound for secret sharing schemes based on qudit graph states. The lower bound on $\alpha$ depends on the dimension $q$ (see Theorem \ref{thm:lower}), the value of  the lower bound is plotted for small values of $q$ in figure \ref{fig:lower}.
 
The lower bound is based on the properties of the \emph{kernel with respect to the dealer} defined as follows:

\begin{definition}
Given a $q$-multigraph $G$, for any $d\in V(G)$ and any $B\subseteq V(G){{\setminus}} \{d\}$, let $\mathcal S_d(B)= \ker(\Gamma_G[B{\cup} \{d\}]) {{\setminus}}  \ker (\Gamma_G[B])$ be the kernel of $B$ with respect to $d$. 
\end{definition}

\begin{lemma} Given a $q$-multigraph $G$, for any $d\in V(G)$ and any $B\subseteq V(G){{\setminus}} \{d\}$,  if $\partial_d\cutrk_G(B) = -1$, there exists $C\in \mathcal S_d(B)$ such that  $$|\supp(C)| <\frac {q}{q+1}\cutrk_G(B)$$
\label{lemma2}
\end{lemma}

\begin{proof} 
Since $\cutrk_G(B{\cup} \{d\})-\cutrk_G(B) = -1$, $\dim(\ker(\Gamma_G[B{\cup} \{d\}])) - \dim( \ker (\Gamma_G[B]) )= 2$. Moreover, $\ker (\Gamma_G[B]) \subseteq \ker (\Gamma_G[B{\cup} \{d\}])$, so $|\mathcal S_d(B)|=(q^2-1).q^t$ where $t = \dim(  \ker($\\$\Gamma_G[B]))$. Let $M = {I\choose M'}$ a matrix in standard form (or reduced column echelon form) generating $\Gamma_G[B{\cup} \{d\}]$.  Since $|\mathcal S_d(B)|=(q^2-1).q^t$ and $|\ker(\Gamma_G[B{\cup} \{d\}])|= q^{t+2}$, there exist two columns $C_1$ and $C_2$ of $M$ such that  $\forall (x,y)\in [0,q-1]^2{{\setminus}} \{(0,0)\}$, $x.C_1+y.C_2\in \mathcal S_d(B)$.  Notice that since $M$ is in standard form, $|\supp(C_1){\cup} \supp(C_2)| \le |B|+1-t$. Moreover for any $v\in \supp(C_1){\cup} \supp(C_2)$, $v$ has a  zero multiplicity in $q-1$ vectors of the $q^2-1$ linear combinations $x.C_1+y.C_2$ for $x,y \in [0,q-1]{{\setminus}} \{(0,0)\}$, so $\sum_{(x,y)\in [0,q-1]^2{{\setminus}} \{(0,0)\}} |\supp(x.C_1+y.C_2)| = (q^2-1-(q-1)).|\supp(C_1){\cup} \supp(C_2)|$, so there exists $C\in \mathcal S_d(B)$ such that $|\supp (C)|\le \frac {q^2-q}{q^2-1}(|B|+1-t)= \frac q{q+1}(\cutrk_G(B)+1)<   \frac q{q+1}\cutrk_G(B)$. 
\end{proof}


\begin{theorem}\label{thm:lower}
If a  $q$-multigraph $G$ of order $n$ realises a $((\alpha n,n))_q$ scheme, then
$$ 
{n\choose \frac{(1-\alpha)qn}{q+1}}{\alpha n \choose (2\alpha-1)n}\ge \frac{(2\alpha -1)(1-\alpha)}{2}{n \choose \alpha n}$$

\noindent Asymptotically, as $n$ tends to infinity, $\alpha$  satisfies:


$$H_2(\frac {\alpha q+1}{q+1}) + \alpha H_2(\frac{1-\alpha} \alpha) \ge H_2(\alpha)$$ 
\end{theorem}

\begin{proof} Given $B_0$ of size $\alpha n$, according to lemma \ref{lemma2} there exists $C_0\in \mathcal S_d(B_0)$ such that $|\supp(C_0)| < \frac q {q+1}(1-\alpha)n$. Notice that the set $\supp(C_0){\cup}\supp(\Gamma_G.C_0)$ has some partial information about the secret so $|\supp(C_0){\cup}\supp(\Gamma_G.C_0)|\ge  (1-\alpha)n$. Moreover for any $B$ of size $\alpha n$, if $C_0\in \mathcal S_q(B)$ then $\supp(C){\cup}\supp(\Gamma_G.C)\subseteq B$. So there are at most ${n-1-(1-\alpha)n\choose\alpha n - (1-\alpha)n } ={\alpha n -1 \choose (2\alpha -1)n}$ sets $B\subseteq V{{\setminus}} \{d\}$ of size $\alpha n$ such that $C_0\in \mathcal S_d(B)$. 
For any $B$ of size $\alpha n$ there is a $C$ which support is of size at most $\frac q {q+1}(1-\alpha)n-1$, any every such $C$ is associated with at most ${\alpha n -1 \choose (2\alpha -1)n}$ such $B$s, so a counting argument implies ${n-1\choose \alpha n}\le{{\alpha n -1 \choose (2\alpha -1)n}}\sum_{i=1}^{\frac q {q+1}(1-\alpha)n-1} {n-1\choose i}$. Moreover, $\sum_{i=1}^{\frac q {q+1}(1-\alpha)n-1} {n-1\choose i}\le \frac{1+\alpha q}{q(2\alpha -1)}{ n-1 \choose \frac {(1-\alpha)qn}{q+1}-1} = \frac{(1-\alpha)(1+\alpha q)}{(2\alpha -1)(q+1)}{ n \choose \frac {(1-\alpha)qn}{q+1}}$. So, 
$\frac{{n\choose \alpha n}}{{\alpha n \choose (2\alpha -1)n}} = \frac{\alpha}{(1-\alpha)^2}\frac{{n-1\choose \alpha n}}{{\alpha n -1 \choose (2\alpha -1)n}}\le \frac{\alpha (1+\alpha q)}{(2\alpha -1)(1-\alpha)(q+1)}$\\
${n\choose \frac{(1-\alpha)qn}{q+1}}\le \frac{2}{(2\alpha -1)(1-\alpha)}{n\choose \frac{(1-\alpha)qn}{q+1}}$.

Since $2^{n(H_2( p) +o(1))}\le {n\choose pn} \le 2^{nH_2(p)}$, asymptotically, as $n$ tends to infinity, $\alpha$ satisfies the equation $H_2(\frac {\alpha q+1}{q+1}) + \alpha H_2(\frac{1-\alpha} \alpha) \ge H_2(\alpha)$.
\end{proof}

\begin{figure}
\center\scalebox{0.7}{
\setlength{\unitlength}{0.240900pt}
\ifx\plotpoint\undefined\newsavebox{\plotpoint}\fi
\sbox{\plotpoint}{\rule[-0.200pt]{0.400pt}{0.400pt}}%
\begin{picture}(1500,900)(0,0)
\sbox{\plotpoint}{\rule[-0.200pt]{0.400pt}{0.400pt}}%
\put(211.0,131.0){\rule[-0.200pt]{4.818pt}{0.400pt}}
\put(191,131){\makebox(0,0)[r]{ 0.5}}
\put(1419.0,131.0){\rule[-0.200pt]{4.818pt}{0.400pt}}
\put(211.0,243.0){\rule[-0.200pt]{4.818pt}{0.400pt}}
\put(191,243){\makebox(0,0)[r]{ 0.501}}
\put(1419.0,243.0){\rule[-0.200pt]{4.818pt}{0.400pt}}
\put(211.0,355.0){\rule[-0.200pt]{4.818pt}{0.400pt}}
\put(191,355){\makebox(0,0)[r]{ 0.502}}
\put(1419.0,355.0){\rule[-0.200pt]{4.818pt}{0.400pt}}
\put(211.0,467.0){\rule[-0.200pt]{4.818pt}{0.400pt}}
\put(191,467){\makebox(0,0)[r]{ 0.503}}
\put(1419.0,467.0){\rule[-0.200pt]{4.818pt}{0.400pt}}
\put(211.0,579.0){\rule[-0.200pt]{4.818pt}{0.400pt}}
\put(191,579){\makebox(0,0)[r]{ 0.504}}
\put(1419.0,579.0){\rule[-0.200pt]{4.818pt}{0.400pt}}
\put(211.0,691.0){\rule[-0.200pt]{4.818pt}{0.400pt}}
\put(191,691){\makebox(0,0)[r]{ 0.505}}
\put(1419.0,691.0){\rule[-0.200pt]{4.818pt}{0.400pt}}
\put(211.0,803.0){\rule[-0.200pt]{4.818pt}{0.400pt}}
\put(191,803){\makebox(0,0)[r]{ 0.506}}
\put(1419.0,803.0){\rule[-0.200pt]{4.818pt}{0.400pt}}
\put(211.0,131.0){\rule[-0.200pt]{0.400pt}{4.818pt}}
\put(211,90){\makebox(0,0){ 0}}
\put(211.0,839.0){\rule[-0.200pt]{0.400pt}{4.818pt}}
\put(416.0,131.0){\rule[-0.200pt]{0.400pt}{4.818pt}}
\put(416,90){\makebox(0,0){ 5}}
\put(416.0,839.0){\rule[-0.200pt]{0.400pt}{4.818pt}}
\put(620.0,131.0){\rule[-0.200pt]{0.400pt}{4.818pt}}
\put(620,90){\makebox(0,0){ 10}}
\put(620.0,839.0){\rule[-0.200pt]{0.400pt}{4.818pt}}
\put(825.0,131.0){\rule[-0.200pt]{0.400pt}{4.818pt}}
\put(825,90){\makebox(0,0){ 15}}
\put(825.0,839.0){\rule[-0.200pt]{0.400pt}{4.818pt}}
\put(1030.0,131.0){\rule[-0.200pt]{0.400pt}{4.818pt}}
\put(1030,90){\makebox(0,0){ 20}}
\put(1030.0,839.0){\rule[-0.200pt]{0.400pt}{4.818pt}}
\put(1234.0,131.0){\rule[-0.200pt]{0.400pt}{4.818pt}}
\put(1234,90){\makebox(0,0){ 25}}
\put(1234.0,839.0){\rule[-0.200pt]{0.400pt}{4.818pt}}
\put(1439.0,131.0){\rule[-0.200pt]{0.400pt}{4.818pt}}
\put(1439,90){\makebox(0,0){ 30}}
\put(1439.0,839.0){\rule[-0.200pt]{0.400pt}{4.818pt}}
\put(211.0,131.0){\rule[-0.200pt]{0.400pt}{175.375pt}}
\put(211.0,131.0){\rule[-0.200pt]{295.825pt}{0.400pt}}
\put(1439.0,131.0){\rule[-0.200pt]{0.400pt}{175.375pt}}
\put(211.0,859.0){\rule[-0.200pt]{295.825pt}{0.400pt}}
\put(30,495){\makebox(0,0){$\alpha$}}
\put(825,29){\makebox(0,0){$q$}}
\put(1279,819){\makebox(0,0)[r]{minimum value of $\alpha$ s.t. $H_2(\frac{\alpha q+1}{q+1}) + \alpha H_2(\frac{1-\alpha}{\alpha}) \ge H_2(\alpha)$}}
\put(1299.0,819.0){\rule[-0.200pt]{24.090pt}{0.400pt}}
\put(293,844){\usebox{\plotpoint}}
\multiput(293.58,828.40)(0.498,-4.608){79}{\rule{0.120pt}{3.759pt}}
\multiput(292.17,836.20)(41.000,-367.199){2}{\rule{0.400pt}{1.879pt}}
\multiput(334.58,464.23)(0.499,-1.314){161}{\rule{0.120pt}{1.149pt}}
\multiput(333.17,466.62)(82.000,-212.616){2}{\rule{0.400pt}{0.574pt}}
\multiput(416.00,252.92)(0.684,-0.499){117}{\rule{0.647pt}{0.120pt}}
\multiput(416.00,253.17)(80.658,-60.000){2}{\rule{0.323pt}{0.400pt}}
\multiput(498.00,192.92)(1.571,-0.497){49}{\rule{1.346pt}{0.120pt}}
\multiput(498.00,193.17)(78.206,-26.000){2}{\rule{0.673pt}{0.400pt}}
\multiput(579.00,166.92)(3.232,-0.493){23}{\rule{2.623pt}{0.119pt}}
\multiput(579.00,167.17)(76.556,-13.000){2}{\rule{1.312pt}{0.400pt}}
\multiput(661.00,153.93)(7.362,-0.482){9}{\rule{5.567pt}{0.116pt}}
\multiput(661.00,154.17)(70.446,-6.000){2}{\rule{2.783pt}{0.400pt}}
\multiput(743.00,147.93)(10.795,-0.488){13}{\rule{8.300pt}{0.117pt}}
\multiput(743.00,148.17)(146.773,-8.000){2}{\rule{4.150pt}{0.400pt}}
\put(907,139.17){\rule{16.500pt}{0.400pt}}
\multiput(907.00,140.17)(47.753,-2.000){2}{\rule{8.250pt}{0.400pt}}
\multiput(989.00,137.93)(45.460,-0.477){7}{\rule{32.820pt}{0.115pt}}
\multiput(989.00,138.17)(340.880,-5.000){2}{\rule{16.410pt}{0.400pt}}
\put(293,844){\makebox(0,0){$+$}}
\put(334,469){\makebox(0,0){$+$}}
\put(416,254){\makebox(0,0){$+$}}
\put(498,194){\makebox(0,0){$+$}}
\put(579,168){\makebox(0,0){$+$}}
\put(661,155){\makebox(0,0){$+$}}
\put(743,149){\makebox(0,0){$+$}}
\put(907,141){\makebox(0,0){$+$}}
\put(989,139){\makebox(0,0){$+$}}
\put(1398,134){\makebox(0,0){$+$}}
\put(1349,819){\makebox(0,0){$+$}}
\put(1398.0,134.0){\rule[-0.200pt]{9.877pt}{0.400pt}}
\put(211.0,131.0){\rule[-0.200pt]{0.400pt}{175.375pt}}
\put(211.0,131.0){\rule[-0.200pt]{295.825pt}{0.400pt}}
\put(1439.0,131.0){\rule[-0.200pt]{0.400pt}{175.375pt}}
\put(211.0,859.0){\rule[-0.200pt]{295.825pt}{0.400pt}}
\end{picture}}
\caption{Lower bound on the accessibility to quantum information in a $((k,n))_q$ scheme. There is no $((k, n))_q$ scheme with $k\leq \alpha n$}
\label{fig:lower}
\end{figure}
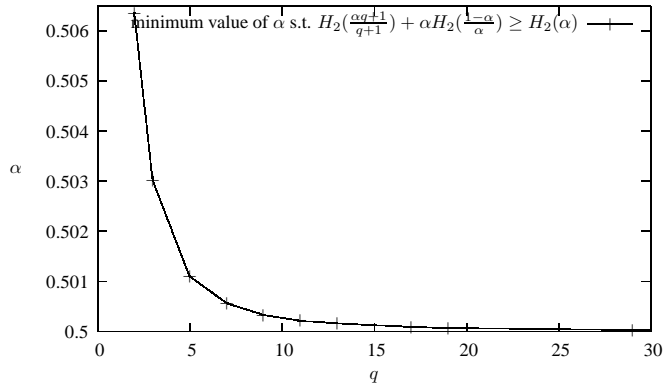

\section{Discussion}

In this work we have studied the encoding of classical and quantum information onto graph states of qudits, and its application for secret sharing schemes. 
We have given complete graphical characterization of which sets of vertices (players) can access the information, and shown how this can be done both for classical and quantum information. Using this characterization we have given bounds on which protocols are possible and how difficult the access structure is to calculate given a graph.

Whilst we have focused on the application of our results for secret sharing, there may be applications to other quantum information protocols. Indeed, the QQ encoding defined in section \ref{Section: encoding} is exactly the same encoding procedure used in measurement based quantum computing and error correction, so we can expect that these results have implications in both these domains. Furthermore, quantum secret sharing is intimately linked to error correction \cite{MM12,CGL99}. All secret sharing schemes are error correcting schemes, and the QQ protocols presented here are equivalent to all possible stabilizer codes \cite{MM12}. Thus, the existence of $((\alpha n, n))$ protocols is an existence statement about error correcting protocols too, and the no goes on secret sharing imply no-goes for all stabilizer codes - so that there are no stabilizer codes with parameters violating our lower bounds.

\section*{Acknowledgements}

The authors want to thank Mehdi Mhalla and David Cattan\'eo for fruitful discussions.  This work has been funded by the ANR-10-JCJC-0208 CausaQ grant, the FREQUENCY (ANR-09-BLAN-0410), HIPERCOM (2011-CHRI-006) projects, and by the Ville de Paris Emergences program, project CiQWii.

\section{Appendix-QQ Encoding-Decoding Operations} \label{Section:appendix}
The QQ encoding-decoding can basically be done by three typical ways. The first method is based on projective Bell measurements (possibly extended to an $|B|+1$ length state) and the two last one are accessible by local measurements and/or implementation of control operations between two qudits, which should  finally result in a similar experimental complexity. We briefly describe the three encoding method $E1, E2, E3$ and decoding $D2, D3$. ($D1$ has been done in section \ref{section:protocols}).  
For a graph $G=(V,\Gamma)$, with $d\in V$, $W:=V{\setminus} \{d\}$, a quantum secret $\ket{\xi}_S:=\sum_{i=0}^{q-1}s_i\ket{i}_S$, we write $\bar{X}:=Z_{\Gamma\{d\}}$ and $\bar{Z}:={(X_uZ_{\Gamma\{u\}})}^{-\Gamma(u,d)^{-1}}$, for $\scriptsize{u, \Gamma(u,d)\neq 0}$, as they act like logical operators over the bases states over $W$, that is $\bar{Z}\ket{i_L}=\omega^i\ket{i_L}$, $\bar{X}\ket{i_L}=\ket{(i+1)_L}$ with notation of \ref{section:protocols}. \\
$\begin{array}{ll}\textbf{E1 } &\ket{\xi}\ket{G}=\sum_{i\in \F_q }s_i\ket{i}_S\sum_{j\in\F_q}\frac{\ket{j}_D\ket{j_L}_W}{\sqrt{q}}\\
 & =\frac{1}{\sqrt{q}}\sum_{i,j\in \F_q }\ket{i}_S\ket{j}_Ds_i\ket{j_L}_W=\frac{1}{\sqrt{q}}\sum_{l\in \F_q}I_SX_D^l\bar{X}_W^l(\sum_{i\in\F_q}\ket{i}_S\ket{i}_Ds_i\ket{i_L}_W)\\
& = \frac{1}{\sqrt{q}}\sum_{l\in \F_q}I_SX_D^l\bar{X}_W^l(\sum_{k\in\F_q}\sum_{i\in\F_q}\omega^{k.i}\frac{\ket{i}_I\ket{i}_D}{q}\sum_j\omega^{-k.j}s_j\ket{j_L}_W)\\
& =\frac{1}{\sqrt{q}}\sum_{l\in \F_q}I_SX_D^l\bar{X}_W^l(\sum_{k\in\F_q}Z_I^kI_D\bar{Z}_W^{-k}\sum_{i\in\F_q}\frac{\ket{i}_I\ket{i}_D}{q}\sum_js_{j\in\F_q}\ket{j_L}_W)\\
& = \frac{1}{q}\sum_{l,k\in \F_q}Z_S^kX_D^l(\sum_{i\in\F_q}\frac{\ket{i}_S\ket{i}_D}{\sqrt{q}})\bar{X}_W^l\bar{Z}_W^{-k}\sum_{j\in\F_q}s_j\ket{j_L}_W)\end{array}$\\
so that applying the correction: $\bar{Z}^{k}\bar{X}^{-l}$ over $V{\setminus} \{d\}$, according to the syndrom $(l,k)$ of a Bell measurement over $\{S,D\}$, leaves the state over $W$ as $\sum_{i\in \F_q} s_i\ket{i_L}$. \\
$\begin{array}{ll}\textbf{E2 } &C\bar{X}_{dW}(\sum_{i\in\F_q}s_i\ket{i})\ket{0_L}=\sum_{i\in\F_q} s_i\ket{i}\bar{X}^{i}\ket{0_L}=\sum_{i\in\F_q} s_i\ket{ii_L}=\sum_{i\in\F_q} s_iX^i\ket{0}\ket{i_L}\\
=&\sum_{i\in\F_q} s_iX^i \sum_{j\in\F_q} \frac{\ket{b_j}}{\sqrt{q}} \ket{i_L}=\sum_{i\in\F_q} s_iX^i \sum_j \frac{Z^{-j}\ket{b_0}}{\sqrt{q}} \ket{i_L}=\frac{1}{\sqrt{q}}\sum_{i,j} s_i\omega^{i.j}Z^{-j}\ket{b_0} \ket{i_L}\\
=&\sum_{j\in\F_q}\frac{\ket{b_j}}{\sqrt{q}}(\bar{Z}^j\sum_{i\in\F_q}s_i\ket{i_L})\end{array}$ \\where $\ket{\,b_k\,}=Z^k\sum_i\ket{i}$ constitutes the $X$ basis ($\ket{b_0}=\ket{+}$).
so that the correction $\bar{Z}^{-i}$ over $W$, according to a $X^i_d$ measurement leaves the state to distribute as $\sum_i s_i \ket{i_L}$. \cite{BCGSZ11}\\
$\begin{array}{ll}\textbf{E3 } &C\bar{Z}_{dW}H_{d}C\bar{X}_{dW}(\sum_{i}s_i\ket{i})\ket{0_L}=C\bar{Z}_{dW}H_{d}(\sum_{i}s_i\ket{i}\ket{i_L})\\
&=C\bar{Z}_{dW}(\sum_{i}s_i\ket{b_i(q)}\ket{i_L})=\ket{+}\sum_{i}s_i\ket{i_L})\end{array}$\\
The same process can be done for the decoding by an authorised set $B$, where the operators $U_B$ and $V_B$ defined in \ref{section:protocols} will act as $\bar{Z}$ and $\bar{X}$ operators respectively. An ancilla qudit $\{a\}$ is prepared in the state $\ket{+}_a$ by $B$.\\
$\textbf{\ \ D2. } {CV^{-1}}_{aB}\ket{+}_a(\sum_{j\in\F_q}s_j\ket{j_L}_W)=\frac{1}{\sqrt{q}}\sum_{k\in\F_q}X^{-k}_a(\sum_{i\in\F_q}s_i\ket{i}_a)\ket{k_L}_W$\\
\ $\textbf{\ \ D3. } CU_{aB}.H_a.\frac{1}{\sqrt{q}}\sum_{k\in\F_q}X^{-k}(\sum_{i\in\F_q}s_i\ket{i}_a)\ket{k_L}_W=\sum_{i\in\F_q}s_i\ket{b_i}_a\sum_{k\in\F_q^m}\frac{\ket{k_L}_W}{\sqrt{q}}$\\
The explicit two qudits gate control operations $B$ has to perform are given by decoding parameters $(x_i,z_i,c)$ defined in \ref{section:protocols}.


\begin{thebibliography}{99}
\small{

\bibitem{BB06} M.~Bahramgiri,S.~Beigi, 
\emph{Graph states under the action of local Clifford group in non-binary case}
arXiv:quant-ph/0610267 (2006).

\bibitem{BCGSZ11} S.~Beigi, I.~Chuang, M.~Grassl, P.~Shor, B.~and Zeng,
  \emph{Graph concatenation for quantum codes}
J. Math. Phys. 52, 022201 (2011).

\bibitem{BCGHS06}    M.~Ben-Or, C.~Cr\'epeau, D.~Gottesman, A.~Hassidim, A.~Smith, 
  \emph{Secure Multiparty Quantum Computation with (Only) a Strict Honest.}
  Proc. 47th Annual IEEE Symposium on the Foundations of Computer Science (FOCS '06), pp. 249-260 (2006).

\bibitem{B09} A.~Broadbent, P.~Chouha, A.~Tapp,
  \emph{The GHZ state in secret sharing and entanglement simulation.}
  Third International Conference on Quantum, Nano and Micro Technologies, ICQNM'09, 59-62 (2009).

\bibitem{B94} A.~Bouchet
  \emph{Circle Graph Obstructions}
  Journal of Combinatorial Theory, Series B, Vol 60, 1,pp 107-144 (1994).
 
\bibitem{BunchHop} J.R. Bunch, J.E. Hopcroft. 
  \emph{Triangular Factorization and Inversion by Fast Matrix Multiplication.} 
  Mathematics of Computation, 28(125):231236, (1974).

\bibitem{CattaneoPerdrix} David Cattan\'eo, Simon Perdrix. 
  \emph{Parametrized Complexity of Weak Odd Domination Problems}. 
  arXiv:1206.4081 (2012).

\bibitem{CGL99} R.~Cleve, D.~Gottesman, H.K.~Lo, 
  \emph{How to share a quantum secret.}
  Phys. Rev. Lett. \textbf{83}, 648 (1999).

\bibitem{Eke91}	A.K.~Ekert,
  \emph{Quantum cryptography based on Bell’s theorem.}
  Phys. Rev. Lett., \textbf{67}, 6, pp. 661--663 (1991).
 

\bibitem{FG11} B.~Fortescue and G.~Gour,
  \emph{Reducing the quantum communication cost of quantum secret sharing.}
  IEEE Trans. Inf. Th. 58(10), pp. 6659 - 6666 (2012)

\bibitem{Gheo12} V.~Gheorghiu, 
  \emph{Generalized Semi-Quantum Secret Sharing Schemes.}
  Phys. Rev. A 85, 052309 (2012)

\bibitem{GLG09} V.~Gheorghiu, S.Y.~Looi, R.B.~ Griffiths,
  \emph{Location of quantum information in additive graph codes}
  Phys. Rev. A, \textbf{81}, 3, pp. 032326, (2010).

\bibitem{Gravier11} S.~Gravier, J.~Javelle, M.~Mhalla and S.~Perdrix, 
  \emph{On Weak Odd Domination and Graph-based Quantum Secret Sharing.}
  arXiv:1112.2495 (2011).

\bibitem{Gravier11b} S.~Gravier, J.~Javelle, M.~Mhalla and S.~Perdrix, 
  \emph{Optimal accessing and non-accessing structures for graph protocols.}
  arXiv:1109.6181 (2011).


\bibitem{HEB03} M. Hein, J. Eisert,  H. J. Briegel. 
\emph{Multiparty entanglement in graph states.} Phys. Rev. A 69, 062311 (2004).

\bibitem{HWE06}M.~Hein, W.~D\"ur, J.~Eisert, R.~Raussendorf, M.~Van~den~Nest, H.~J.~Briegel, 
\emph{Entanglement in graph states and its applications}
in Quantum Computers, Algorithms and Chaos, Proceedings of the International School of Physics Enrico Fermi, Vol. 162 (2006).


\bibitem{HBB99}	M. Hillery, V. Bu\v{z}ek, and A. Berthiaume,
  \emph{quantum secret sharing.}
  Phys. Rev. A \textbf{59}, 1829 (1999).
  
\bibitem{IbarraMH} O.H. Ibarra, S. Moran, R. Hui. 
  \emph{A Generalization of the Fast LUP Matrix Decomposition Algorithm and Applications.} 
  Journal of Algorithms, 3(1):4532656, (1982). 
  
\bibitem{JMP11} J.~Javelle, M.~Mhalla, S.~Perdrix, 
  \emph{New Protocols and Lower Bound for Quantum Secret Sharing with Graph States.}
  TQC'12. LNCS Vol 7582, pp 1-12 (2013).

\bibitem{Javelle12} J. Javelle, M. Mhalla, S. Perdrix,  
  \emph{On the Minimum Degree up to Local Complementation: Bounds and Complexity.}
  WG'12. LNCS Vol 7551, pp 138-147 (2012).

\bibitem{KanteRao}M.M.~Kant\'e, M. Rao.
  \emph{The Rank-Width of Edge-Coloured Graphs}. 
  Theory of Computing Systems, 1-46, (2012).
  
\bibitem{KMMP09} E.~Kashefi, D.~Markham, M.~Mhalla, and S.~Perdrix.
  \emph{Information flow in Secret Sharing Protocols.}
  {\em Electronic Proceedings in Theoretical Computer Science},
  9:87--97 (2009).

\bibitem{KKI99} A.~Karlsson, M.~Koashi, N.~Imoto, 
  \emph{Quantum entanglement for secret sharing and secret splitting.}
       {\sl Phys. Rev. A} \textbf{59}, 162--168, (1999).
  
\bibitem{Keet10} A.~Keet, B.~Fortescue, D.~Markham and B.~C.~Sanders, 
  \emph{Quantum secret sharing with qudit graph states.}
  Phys. Rev. A \textbf{82}, 062315 (2010).

\bibitem{MM12} A.~Marin, D.~Markham,
  \emph{On the equivalence between sharing quantum and classical secrets, and error correction.}
  arxiv:1205.4182 (2012)

\bibitem{MM13} A.~Marin, D.~Markham, 
  \emph{High dimensional CSS code and application to secret sharing.}
  in preparation (2013).

  \bibitem{MS08} D.~Markham and B.~C. Sanders, 
  \emph{Graph State for Quantum Secret Sharing.}
  Phys. Rev. A, \textbf{78}, (2008).

\bibitem{OumSeymour} S. Oum, P. Seymour. 
  \emph{Approximating rank-width and clique-width quickly}. Journal
  ACM Transactions on Algorithms (TALG), vol 5, 1-20 (2008). 
  
\bibitem{S} A.~Shamir,
  \emph{How to share a secret}.
  Communications of the ACM, \textbf{22}, 612–613 (1979).

\bibitem{SS10} L.~Sheridan and V.~Scarani. 
  \emph{Security proof for quantum key distribution using qudit systems.} 
  Phys. Rev. A \textbf{82}, 030301(R) (2010).

\bibitem{VdN04}M.~Van~den~Nest and J.~Dehaene, B. De Moor 
  \emph{Graphical description of the action of local Clifford transformations on graph states}. 
  Physical Review A (69) 022316 (2004). 

\bibitem{KKKS05} A.~Ketkar, A.~Klappenecker, S.~Kumar, P.~K.~Sarvepalli, 
  \emph{Nonbinary Stabilizer Codes Over Finite Fields }
  IEEE Trans. Inf. Th. 52, 4892 (2005).

}\end{thebibliography}
\end{document}